%% file: GeoVlasov-V10.tex
\documentclass[11pt]{amsart}
\include{preamble}
\usepackage{tikz-cd}

\begin{document}
\title{Matched Pair Analysis of Euler-Poincar\'{e} Flow on Hamiltonian Vector Fields}
\author{Oğul Esen}
\address{Department of Mathematics, Gebze Technical University,  41400 Gebze-Kocaeli, Turkey}
\email{oesen@gtu.edu.tr}

\author{Cristina Sardon Mu\~noz}
\address{Department of Applied Mathematics
Universidad Polit\'ecnica de Madrid
C/ Jos\'e Guti\'errez Abascal, 2, 28006, Madrid. Spain}
\email{mariacristina.sardon@upm.es
}

\author{Marcin Zajac}
\address{Department of Mathematical Methods in Physics, University of Warsaw, ul. Pasteura 5, 02-093 Warsaw, Poland.}
\email{marcin.zajac@fuw.edu.pl}
\date{}

\maketitle

\begin{abstract}
In this paper we provide a matched pair decomposition of the space of symmetric contravariant tensors $\mathfrak{T}\mathcal{Q}$. From this procedure two complementary Lie subalgebras of $\mathfrak{T}\mathcal{Q}$ under \textit{mutual} interaction arise. Introducing a lift operator, the matched pair decomposition of the space of Hamiltonian vector fields is determined. According to these realizations, Euler-Poincar\'{e} flows on such spaces are decomposed into two subdynamics: one of which is the Euler--Poincar\'{e} formulation of isentropic fluid flows, and the other one corresponds with Euler--Poincar\'{e} equations on higher order contravariant tensors ($n\geq 2$). \\[0.2cm]
 \textsl{MSC}: 17B66, 37K30. \\
 \textsl{Keywords}: Matched pair Lie algebras; symmetric contravariant tensors; Hamiltonian vector fields; Euler-Poincar\'{e} equations.
\end{abstract}

\maketitle
\tableofcontents
\setlength{\parindent}{0cm}
\setlength{\parskip}{1em}
\onehalfspace

\section{Introduction}

The modern geometric approach of classical mechanics has considered Lagrangian dynamics on a geometric framework defined over tangent bundles. Here, the positions and the velocities of the dynamical system are gathered in the tangent space $T\mathcal{Q}$ of the configuration manifold $\mathcal{Q}$ of the generalized position coordinates \cite{de2011methods,  holm2008geometric, libermann2012symplectic, Ol93}. If the system possesses a Lie group symmetry, then a reduction procedure is possible and one arrives at a reduced system with reduced dynamics (i.e. reduced number of degrees of freedom) on the quotient space \cite{marsden1993lagrangian}. This specific procedure is known as Lagrangian reduction \cite{n2001lagrangian}. If the configuration space is a Lie group $G$ and the Lagrangian function on the tangent bundle $TG$ governing the dynamics is invariant under the group action, then the reduction reads a Lagrangian on the Lie algebra $\mathfrak{g}$. This process is called Euler-Poincar\'{e} reduction, whereas the reduced dynamics is given by the so-called Euler-Poincar\'{e} equations \cite{holm1998euler,MarsdenRatiu-book}. Many physical theories fit this geometric background, e.g., the theory of the rigid body and fluid and plasma theories \cite{abraham1978foundations, arnold1989mathematical, marsden1982group,vizman2008geodesic}.

The problem addressed in this work is the decoupling of Euler-Poincar\'{e} dynamics available on the space of Hamiltonian vector fields. This can be regarded as the Lagrangian realization of the Vlasov dynamics of plasma particles. To observe this, we start with a bunch of plasma particles at rest on $\mathcal{Q}\subset \mathbb{R}^3$. Now, consider the cotangent bundle denoted by $T^*\mathcal{Q}$, which is by definition the dual of $T\mathcal{Q}$ and it admits a canonical symplectic two form. The configuration group of Vlasov's plasma is the group ${\rm Diff_{can}}(T^*\mathcal{Q})$ of  diffeomorphisms 
preserving the canonical form, and this group is infinite dimensional. The Lie algebra can be identified with the space of Hamiltonian vector fields $\mathfrak{X}_{ham}(T^*\mathcal{Q})$ equipped with the minus of the Jacobi-Lie bracket. The minus sign is the manifestation of the symmetry due to the right action of the group \cite{holm1998euler}. 
Physically, the right symmetry corresponds with the particle relabelling of plasma particles. In \cite{holm2009geodesic}, a quadratic Lagrangian function is introduced  on $\mathfrak{X}_{ham}(T^*\mathcal{Q})$ by means  of an operator so that a geodesic flow is obtained.  

One can identify $\mathfrak{X}_{ham}(T^*\mathcal{Q})$ with the space of (smooth) functions modulo constants. In this case, the dual space of the Lie algebra turns out to be the space of densities on $T^*\mathcal{Q}$ \cite{marsden1983hamiltonian}. This identification enables us to recast the Vlasov  motion into Hamiltonian (Lie-Poisson) form \cite{MaWe81}. In a recent study \cite{EsSu21}, it was shown that the Vlasov Lie-Poisson equation can be decoupled into two subdynamics under mutual interactions in a non-trivial way. The geometry introduced in this study is the theory of matched pair Lie algebras,
which permits to (de)couple two Lie algebras under mutual interaction \cite{Maji90,Maji90-II,Majid-book}. Such analysis is beyond the semidirect theory, since there is only a one-sided action in this theory \cite{cendra1998lagrangian,
	n2001lagrangian, holm1998euler, esen2014tulczyjew,MarsRatiWein84}. The abstract framework is the Hamiltonian matched pair theory given in \cite{EsSu16}, and referring to that decomposition, it is shown in \cite{EsSu21} that the constitutive subdynamics of the Hamiltonian formulation of the motion of plasma is  obtained as isentropic compressible fluid motion and the dynamics of the kinetic moments of the plasma density function of order $n\geq 2$. Therefore, such an algebraic/geometric decomposition of the Hamiltonian Vlasov theory is consistent with the physical intuition. 
	
In this study, our aim is to present a similar decomposition for the geodesic Vlasov motion \cite{holm2009geodesic} in terms of the Euler-Poincar\'{e} formulation and obtain a matched pair decomposition of this system. The abstract framework of  Lagrangian matched pair theory is already available in the literature \cite{EsenSutl17} and it provides matched Euler-Poincar\'{e} equations for classical Lie algebras. This paper achieves to provide a physical example to shed some light over this algebraic construction. 
	
\textbf{Notation.} We denote the Lie algebras by $\mathfrak{g}$, $\mathfrak{h}$ and $\mathfrak{K}$. Throughout the work, we shall adopt the letters%
\begin{equation*}\label{G}
\xi,\tilde{\xi} \in \mathfrak{g},\qquad \mu,\tilde{\mu}
 \in\mathfrak{g}^{\ast} , \qquad \eta,\tilde{\eta}, \in \mathfrak{g},\qquad \nu,\tilde{\nu}
 \in\mathfrak{h}^{\ast}, \qquad x,\tilde{x} \in \mathfrak{K}
\end{equation*}
as elements of the given spaces. We shall use the notation $ad_{\xi}\tilde{\xi}=[\xi,\tilde{\xi}]$ for the infinitesimal left adjoint representation and we denote the infinitesimal coadjoint action of $\mathfrak{g}$ on $\mathfrak{g}^{\ast}$ by $%
ad_{\xi}^{\ast}$. The latter is defined to be the minus of the linear algebraic dual of $ad_\xi$, that is $\langle ad_{\xi}^{\ast}\mu,\tilde{\xi}\rangle =-\langle
\mu,ad_{\xi}\tilde{\xi} \rangle$ for all $\xi, \tilde{\xi}$ in $\mathfrak{g}$ and $\mu\in\mathfrak{g}^{\ast}$.

\section{Coupling of Euler-Poincar\'{e} Dynamics} \label{Sec-mp}

In this subsection we recall the basics on matched pairs of Lie groups, Lie algebras, and Lie coalgebras. There is an extensive literature on these subjects, see e.g. \cite{LuWe90,Maji90-II,Maji90,Majid-book,Mich80,Ta81,Zhan10} for further details. We also refer the reader to \cite{EsSu16,EsenSutl17, esen2018matched}.

\subsection{Matched Pair Lie Algebras}\label{Sec-mpla}~

A matched pair Lie algebra $\mathfrak{g}\bowtie \mathfrak{h}$ is a Lie algebra containing two non-intersecting Lie algebras $\mathfrak{g}$ and $\mathfrak{h}$ under mutual interaction. We represent the mutual actions by
\begin{equation} \label{Lieact}
\vartriangleright:\mathfrak{h}\otimes\mathfrak{g}\rightarrow \mathfrak{g},\quad \eta\otimes \xi \mapsto \eta\vartriangleright \xi, \qquad  \vartriangleleft:\mathfrak{h}\otimes\mathfrak{g}\rightarrow \mathfrak{h}, \quad \eta\otimes \xi \mapsto \eta\vartriangleleft \xi.
\end{equation}
In this case, the Lie algebra bracket defined on $\mathfrak{g}\bowtie \mathfrak{h}$ is given by 
\begin{equation}\label{mpla}
\lbrack  \xi \oplus \eta ,\, \tilde{\xi}\oplus\tilde{\eta} ]=\big( [\xi,\tilde{\xi}]+\eta\vartriangleright \tilde{\xi}-\tilde{\eta}\vartriangleright \xi\big) \oplus \big(
[\eta,\tilde{\eta}]+\eta\vartriangleleft \tilde{\xi}-\tilde{\eta}\vartriangleleft \xi\big).  
\end{equation}
The Jacobi identity for the matched pair Lie algebra bracket manifests the following compatibility conditions
\begin{equation} \label{compcon-mpl}
\begin{split}
\eta \vartriangleright \lbrack \xi,\tilde{\xi}]=[\eta \vartriangleright
\xi,\tilde{\xi}]+[\xi,\eta \vartriangleright \tilde{\xi}]+(\eta
\vartriangleleft \xi)\vartriangleright \tilde{\xi}-(\eta \vartriangleleft
\tilde{\xi})\vartriangleright \xi, \\
\lbrack \eta,\tilde{\eta}]\vartriangleleft\xi =[\eta,\tilde{\eta}\vartriangleleft\xi ]+[\eta
\vartriangleleft\xi ,\tilde{\eta}]+\eta\vartriangleleft (\tilde{\eta}\vartriangleright \xi
)-\tilde{\eta}\vartriangleleft (\eta\vartriangleright \xi ).
\end{split}
\end{equation}
Now, we present here the following proposition, which will be useful in upcoming sections. We refer the reader to \cite[Prop. 8.3.2]{Majid-book} for further details. 

\begin{proposition} \label{universal-prop}
	Let $\mathfrak{s}$ be a Lie algebra with two Lie subalgebras $\mathfrak{g}$ and $\mathfrak{h}$ such that $\mathfrak{s}$ is isomorphic to the direct sum of $\mathfrak{g}$ and $\mathfrak{h}$ as vector spaces with the vector addition property in $\mathfrak{s}$. Then $\mathfrak{s}$ is isomorphic to the matched pair $\mathfrak{g}\bowtie\mathfrak{h}$ as Lie algebras, and the mutual actions are derived from 
	\begin{equation} \label{mab-defn}
	[\eta,\xi]=(\eta\rtimes \xi )\oplus( \eta\ltimes \xi).
	\end{equation}
	Here, the inclusions of the subalgebras are defined to be 
	\begin{equation}
	\mathfrak{g} \longrightarrow \mathfrak{s}: \xi \mapsto (\xi\oplus 0),\qquad \mathfrak{h} \longrightarrow \mathfrak{s}: \eta \mapsto (0\oplus \eta).
	\end{equation}
\end{proposition}

Now, recall the Lie algebra actions in (\ref{Lieact}) and take the left action $\vartriangleright$ and fix an element $\eta$ in $\mathfrak{h}$: this results in a linear mapping, denoted by $\eta \vartriangleright$, on $\mathfrak{g}$ and the linear algebraic dual of this mapping which reads
\begin{equation} \label{eta-star}
\overset{\ast }{\vartriangleleft} \eta:\mathfrak{g}^*\mapsto \mathfrak{g}^*, 
\qquad \langle \mu \overset{\ast }{\vartriangleleft} \eta, \xi \rangle=\langle \mu, \eta \vartriangleright \xi \rangle.
\end{equation}
This is the right representation of $\mathfrak{h}$ on $\mathfrak{g}^*$. 
On the other hand, by fixing $\xi\in \mathfrak{g}$, we define a linear mapping $\mathfrak{b}_\xi$ from $\mathfrak{h}$ to $\mathfrak{g}$ as
\begin{equation} \label{b}
\mathfrak{b}_\xi: \mathfrak{h} \mapsto \mathfrak{g},\qquad \mathfrak{b}_\xi(\eta)=\eta\vartriangleright \xi.
\end{equation}
The dual of this mapping is
\begin{equation} \label{b*}
\mathfrak{b}_\xi^*:\mathfrak{g}^*\mapsto \mathfrak{h}^*, \qquad \langle \mathfrak{b}_\xi^*\mu,\eta \rangle= \langle \mu, \mathfrak{b}_\xi \eta \rangle = \langle \mu, \eta\vartriangleright \xi  \rangle.
\end{equation} 

Similarly, recall the right action in $\vartriangleleft$ and freeze $\xi$ in $\mathfrak{g}$. This reads a linear mapping, denoted by $\vartriangleleft\xi $, on the Lie algebra $\mathfrak{h}$. The dual of this mapping is 
\begin{equation} \label{xi-star}
\xi \overset{\ast }{\vartriangleright}: \mathfrak{h}^* \mapsto \mathfrak{h}^*, \qquad 
\langle \xi \overset{\ast }{\vartriangleright}\nu, \eta \rangle 
=\langle\nu,  \eta \vartriangleleft\xi\rangle.
\end{equation}
This is the  left representation of $\mathfrak{h}$ on $\mathfrak{g}^*$. Now, we freeze an element, say $\eta$ in $\mathfrak{h}$, in the right action $\vartriangleleft$. This enables us to define a linear mapping $\mathfrak{a}_\eta$ from $\mathfrak{g}$ to $\mathfrak{h}$ that is 
\begin{equation} \label{a}
\mathfrak{a}_\eta:\mathfrak{g}\mapsto \mathfrak{h}, \qquad \mathfrak{a}_\eta(\xi)=\eta\vartriangleleft \xi
\end{equation}
along with the dual mapping
\begin{equation} \label{a*}
\mathfrak{a}_\eta^*:\mathfrak{h}^*\mapsto \mathfrak{g}^*, \qquad \langle \mathfrak{a}_\eta^* \nu,\eta \rangle =
\langle \nu,\mathfrak{a}_\eta \xi \rangle=\langle \nu,\eta\vartriangleleft \xi \rangle.
\end{equation}

\begin{proposition} \label{ad-*-prop}
The infinitesimal coadjoint
action $\ad^{\ast}$ of an element $(\xi\oplus\eta)$ in $\mathfrak{g}\bowtie \mathfrak{h}$ onto an element $(\mu\oplus\nu)$ in the dual space $(\mathfrak{g}^*\oplus \mathfrak{h})^{\ast}$ is computed to be
\begin{equation}  \label{ad-*}
\ad_{(\xi\oplus\eta)}^{\ast}(\mu\oplus\nu)=\underbrace{ \big(ad^{\ast}_{\xi} \mu -\mu \overset{\ast }{%
	\vartriangleleft}\eta - \mathfrak{a}_{\eta}^{\ast}\nu\big)}_{\in ~ \mathfrak{g}^*}\oplus \underbrace{ \big(
ad^{\ast}_{\eta} \nu +\xi \overset{\ast }{\vartriangleright}\nu+ \mathfrak{b}%
_{\xi}^{\ast}\mu\big )}_{\in ~ \mathfrak{h}^*}.
\end{equation}
Here, (the italic) $ad^{\ast}$ represents the infinitesimal coadjoint actions of Lie subalgebras to their duals.
\end{proposition}
\begin{proof}
One way, and probably the easiest one, to prove this proposition is to use  the direct definition of the coadjoint action. Then by employing the matched Lie algebra bracket \eqref{mpla} in this definition, the result follows from a direct calculation \cite{EsSu16}. Instead of presenting this calculation, we prefer to prove the proposition in a different, although a bit longer way, since we found this alternative proof more appropriate for  future aspects of the present paper. One advantage of this proof is to derive the distinction between two notations $\ad_\xi \mu$ (coadjoint action on the matched pair) and $ad_\xi \mu$ (coadjoint action on the Lie subalgebra level). In accordance with this, we start by rewriting the left hand side of \eqref{ad-*} in the following form
\begin{equation} \label{ad*-1-}
\ad_{\xi+\eta}^{\ast}(\mu+\nu)=\ad_{\xi}^{\ast}\mu+\ad_{\xi}^{\ast}\nu+\ad_{\eta}^{\ast}\mu+\ad_{\eta}^{\ast}\nu,
\end{equation}
where, for example, $\xi$ denotes the element $\xi+0$ in $\mathfrak{g}\bowtie \mathfrak{h}$. 
We compute all these terms one by one in respective order to arrive at the right hand side of \eqref{ad-*}. Notice that all four terms on the right hand side are the coadjoint actions of $\mathfrak{g}\bowtie \mathfrak{h}$ on its dual space $(\mathfrak{g}\bowtie \mathfrak{h})^{\ast}$. So that we should couple these terms with  generic Lie algebra elements, say $\tilde{\xi}+\tilde{\eta}$, in $\mathfrak{g}\bowtie \mathfrak{h}$. For the first term in the right hand side of \eqref{ad*-1-}, we compute 
\begin{equation}
\begin{split}
\langle \ad_{\xi}^{\ast}\mu, \tilde{\xi}+\tilde{\eta} \rangle 
&=
\langle \mu, [\tilde{\xi}+\tilde{\eta},\xi] \rangle 
=
\langle \mu, [\tilde{\xi},\xi] \rangle +\langle \mu, [\tilde{\eta},\xi] \rangle
\\
&= \langle  ad_{\xi}^{\ast}\mu, \tilde{\xi} \rangle 
+
\langle \mu, \tilde{\eta}\vartriangleright \xi \rangle 
+
\langle \mu, \tilde{\eta}\vartriangleleft \xi\rangle
\\
&=\langle  ad_{\xi}^{\ast}\mu, \tilde{\xi} \rangle 
+0+
\langle \mathfrak{b}^*_\xi \mu ,\tilde{\eta} \rangle,
\end{split}
  \end{equation}
  where we have employed \eqref{mab-defn} in the bracket $[\tilde{\eta},\xi]$ in the second line. 
This computation shows that the projection of $\ad_{\xi}^{\ast}\mu$ on the dual space $\mathfrak{g}^*$ is $ad_{\xi}^{\ast}\mu$ whereas the projection of $\ad_{\xi}^{\ast}\mu$ on the dual space $\mathfrak{h}^*$ is $\mathfrak{b}^*_\xi \mu $. More formally, we write this as 
\begin{equation} \label{result-1}
\ad_{\xi}^{\ast}\mu=\big(ad_{\xi}^{\ast}\mu \oplus \mathfrak{b}^*_\xi \mu \big) \in \mathfrak{g}^*\oplus \mathfrak{h}^*.
\end{equation} 
In other words, we conclude that $ad_{\xi}^{\ast}\mu$ is the restriction of $\ad_{\xi}^{\ast}\mu$ to $\mathfrak{g}^*$. Next, we study the second term on the right hand side of \eqref{ad*-1-}. Accordingly, for arbitrary $\tilde{\xi}+\tilde{\eta}$ in $\mathfrak{g}\bowtie \mathfrak{h}$ we have 
\begin{equation}
\begin{split}
\langle \ad_{\xi}^{\ast}\nu, \tilde{\xi}+\tilde{\eta} \rangle 
&=
\langle \nu, [\tilde{\xi}+\tilde{\eta},\xi] \rangle 
=
\langle \nu, [\tilde{\xi},\xi] \rangle +\langle \nu, [\tilde{\eta},\xi] \rangle
\\
&= \langle \nu, [\tilde{\xi},\xi] \rangle + \langle \nu, \tilde{\eta}\vartriangleright \xi \rangle  + \langle \nu, \tilde{\eta}\vartriangleleft \xi\rangle \\&= 0+0+\langle \xi \overset{\ast }{\vartriangleright} \nu, \tilde{\eta} \rangle.
\end{split}
  \end{equation}
The first and the second term in the second line are zero since all the possible pairings between $\mathfrak{h}^*$ and $\mathfrak{g}$ vanish. Therefore, we obtain that  
  \begin{equation}\label{result-2}
  \ad_{\xi}^{\ast}\nu= \big(0\oplus \xi \overset{\ast }{\vartriangleright} \nu\big)\in \mathfrak{g}^*\oplus \mathfrak{h}^*.
\end{equation}   
For the third term on the right hand side of \eqref{ad*-1-} we compute 
\begin{equation}
\begin{split}
\langle \ad_{\eta}^{\ast}\mu, \tilde{\xi}+\tilde{\eta} \rangle 
&=
\langle \mu, [\tilde{\xi}+\tilde{\eta},\eta] \rangle 
=
\langle \mu, [\tilde{\xi},\eta] \rangle +\langle \mu, [\tilde{\eta},\eta] \rangle
\\
&= -\langle \mu, \eta\vartriangleright\tilde{\xi} \rangle  
-\langle \mu, \eta\vartriangleleft\tilde{\xi} \rangle  + \langle \mu, [\tilde{\eta},\eta] \rangle \\ &= - \langle \mu  \overset{\ast }{\vartriangleleft} \eta \rangle-0+0.
\end{split}
\end{equation}
So that, we record this as 
\begin{equation} \label{result-3}
\ad_{\eta}^{\ast}\mu=\big (- \mu  \overset{\ast }{\vartriangleleft} \eta\oplus 0\big)\in \mathfrak{g}^*\oplus \mathfrak{h}^*.
\end{equation}
Finally, we work on the fourth term on the right hand side of \eqref{ad*-1-}. We have
\begin{equation}
\begin{split}
\langle \ad_{\eta}^{\ast}\nu, \tilde{\xi}+\tilde{\eta} \rangle 
&=
\langle \nu, [\tilde{\xi}+\tilde{\eta},\eta] \rangle 
=
\langle \nu, [\tilde{\xi},\eta] \rangle +\langle \nu, [\tilde{\eta},\eta] \rangle
\\
&= -\langle \nu, \eta\vartriangleright\tilde{\xi} \rangle  
-\langle \nu, \eta\vartriangleleft\tilde{\xi} \rangle  + \langle \nu, [\tilde{\eta},\eta] \rangle \\&= -0- \langle \mathfrak{a}^*_\eta\nu, \tilde{\xi} \rangle +\langle ad_\eta\nu,\tilde{\eta}\rangle,
\end{split}
\end{equation}
which means that
\begin{equation}\label{result-4}
\ad_{\eta}^{\ast}\nu=\big(-\mathfrak{a}^*_\eta\nu \oplus ad_\eta\nu \big)\in \mathfrak{g}^*\oplus \mathfrak{h}^*.
\end{equation}
Eventually, by adding all the results in \eqref{result-1}, \eqref{result-2}, \eqref{result-3} and \eqref{result-4} we arrive at the right hand side of \eqref{ad-*}. 
\end{proof}  

\subsection{Lie Algebra Homomorphisms}~
\begin{lemma}\label{mp-homo}
Given two matched pair Lie algebras $\mathfrak{K}_1=\mathfrak{g}_1\bowtie\mathfrak{h}_1$ and $\mathfrak{K}_2=\mathfrak{g}_2\bowtie\mathfrak{h}_2$, a linear map $\varphi:\mathfrak{K}_1\mapsto \mathfrak{K}_2$ satisfying $\varphi(\mathfrak{g}_1)\subseteq \mathfrak{g}_2$ and $\varphi(\mathfrak{h}_1)\subseteq \mathfrak{h}_2$ is a Lie algebra homomorphism, if and only if 
\begin{equation} \label{Lie-hom-eq}
\varphi(\eta\vartriangleright \xi)=\varphi(\eta) \vartriangleright \varphi(\xi) \qquad 
\varphi(\eta\vartriangleleft \xi)=\varphi(\eta) \vartriangleleft \varphi(\xi)
\end{equation}
for any $\xi\in \mathfrak{g}_1$ and any $\eta\in \mathfrak{h}_1$.
\end{lemma}
Assume that $\varphi$ is a homomorphism from a matched pair $\mathfrak{K}_1=\mathfrak{g}_1\bowtie\mathfrak{h}_1$ to another matched pair $\mathfrak{K}_2=\mathfrak{g}_2\bowtie\mathfrak{h}_2$. Let us also assume that this homomorphism respects the matched pair decompositions, that is, Lemma \ref{mp-homo} is assumed. Then a straightforward calculation shows that 
\begin{equation} \label{dual-inc}
\varphi^*(\mathfrak{g}_2^*)\subset \mathfrak{g}_1^*, \qquad \varphi^*(\mathfrak{h}_2^*)\subset \mathfrak{h}_1^*,
\end{equation}
where $\varphi^*$ is the dual operation. In the following two lemmas we exhibit commutation rules of  the dual and the cross actions with the dual mapping.
\begin{lemma}\label{b-a-*}
Assume a Lie algebra $\varphi$ as described in Lemma \ref{mp-homo}. Then the commutation rules  
\begin{equation} \label{dual-b-a}
\varphi^*\circ \mathfrak{b}^*_{\varphi(\xi)}=\mathfrak{b}^*_{\xi}\circ \varphi^*, \qquad \varphi^*\circ \mathfrak{a}^*_{\varphi(\eta)}=\mathfrak{a}^*_{\eta}\circ \varphi^*
\end{equation}
hold for the dual mapping   $\varphi^*$ and the cross actions $\mathfrak{b}^*$ in \eqref{b} and $\mathfrak{a}^*$ in \eqref{a*}, respectively. 
\end{lemma}
\begin{proof}
We start with pairing of $\varphi(\eta \vartriangleright \xi)$ with an arbitrary element $\tilde{\mu}$ in $\mathfrak{g}^*_2$, i.e.
\begin{equation} \label{calc--1}
\langle \varphi(\eta \vartriangleright \xi),\tilde{\mu}\rangle
= \langle \varphi(\eta) \vartriangleright \varphi(\xi),\tilde{\mu}\rangle = \langle \varphi(\eta),\mathfrak{b}_{\varphi(\xi)}\tilde{\mu}\rangle = \langle \eta,\varphi^*\circ \mathfrak{b}_{\varphi(\xi)}\tilde{\mu}\rangle.
\end{equation}
Notice that we have employed identity (\ref{Lie-hom-eq}) in the first equality and then we have used the definition of $\mathfrak{b}^*$ from \eqref{b*} in the second equality. On the other hand, we have that
\begin{equation} \label{calc--2}
\langle \varphi(\eta \vartriangleright \xi),\tilde{\mu}\rangle
= \langle \eta \vartriangleright \xi,\varphi^*(\tilde{\mu})\rangle
=\langle \eta ,\mathfrak{b}_\xi \circ \varphi^*(\tilde{\mu})\rangle.
\end{equation}
Comparing the calculations in \eqref{calc--1} and \eqref{calc--2}, we arrive at the first identity in \eqref{dual-b-a} for an arbitrary $\tilde{\mu}$. For the second identity, start with pairing $\varphi(\eta \vartriangleleft \xi)$ with an arbitrary element $\tilde{\nu}$ in $\mathfrak{h}^*_2$. We have that 
\begin{equation} \label{calc--3}
\langle \varphi(\eta \vartriangleleft \xi),\tilde{\nu}\rangle
= \langle \varphi(\eta) \vartriangleleft \varphi(\xi),\tilde{\nu}\rangle = \langle \varphi(\xi),\mathfrak{a}_{\varphi(\eta)}\tilde{\nu}\rangle = \langle \xi,\varphi^*\circ \mathfrak{a}_{\varphi(\eta)}\tilde{\nu}\rangle,
\end{equation}
where identity (\ref{Lie-hom-eq}) is used in the first equality, and definition \eqref{a*} is employed in the second equality. On the other hand, one has
\begin{equation} \label{calc--4}
\langle \varphi(\eta \vartriangleleft \xi),\tilde{\nu}\rangle
= \langle \eta \vartriangleleft \xi,\varphi^*(\tilde{\nu})\rangle
=\langle \xi ,\mathfrak{a}_\eta \circ \varphi^*(\tilde{\nu})\rangle.
\end{equation}
A comparison between  \eqref{calc--3} and \eqref{calc--4} for an arbitrary $\tilde{\nu}$ results in the second identity in \eqref{dual-b-a}.
\end{proof}

\begin{lemma} \label{dual-b-a-}
Assume a Lie algebra $\varphi$ as described in Lemma \ref{mp-homo}. For the dual mapping   $\varphi^*$ and the dual actions $\overset{\ast }{\vartriangleleft}$ in \eqref{eta-star} and $\overset{\ast}{\vartriangleright}$ in \eqref{xi-star}, 
the following commutation rules hold 
\begin{equation} \label{dual-triangle-eta-xi}
(\varphi^*\tilde{\mu}) \overset{\ast }{\vartriangleleft} \eta =\varphi^*\big(\tilde{\mu} \overset{\ast }{\vartriangleleft} \varphi(\eta)\big), 
\qquad 
\xi \overset{\ast }{\vartriangleright}\varphi^* (\tilde{\nu})=\varphi^*(\varphi(\xi)\overset{\ast }{\vartriangleright} \tilde{\nu}\big)
\end{equation}
for arbitrary $\tilde{\mu}$ in $\mathfrak{g}_2^*$ and $\tilde{\nu}$ in $\mathfrak{h}_2^*$, respectively. 
\end{lemma}
\begin{proof}
We couple $\varphi(\eta \vartriangleright \xi)$ with an arbitrary element $\tilde{\mu}$ in $\mathfrak{g}^*_2$ that is,
\begin{equation} \label{calc--11}
\begin{split}
\langle \varphi(\eta \vartriangleright \xi),\tilde{\mu}\rangle
&=
\langle \varphi(\eta) \vartriangleright \varphi(\xi),\tilde{\mu}\rangle 
=
\langle \varphi(\xi), \tilde{\mu}\overset{\ast }{\vartriangleleft}\varphi(\eta) \rangle 
\\&=
\langle \xi,\varphi^*\big(\tilde{\mu}\overset{\ast }{\vartriangleleft}\varphi(\eta)\big) \rangle,
\end{split}
\end{equation}
where we have employed the identity (\ref{Lie-hom-eq}) in the first equality whereas we have used the definition of $\overset{\ast }{\vartriangleleft}$ in \eqref{eta-star} in the second equality. On the other hand, we have that
\begin{equation} \label{calc--22}
\langle \varphi(\eta \vartriangleright \xi),\tilde{\mu}\rangle
= \langle \eta \vartriangleright \xi,\varphi^*(\tilde{\mu})\rangle
=\langle \xi , \big(\varphi^*(\tilde{\mu})\big)\overset{\ast }{\vartriangleleft}\eta \rangle.
\end{equation}
Comparing the calculations in \eqref{calc--11} and \eqref{calc--22}, one arrives at the first identity in \eqref{dual-triangle-eta-xi} for an arbitrary $\tilde{\mu}$. For the second identity, start with pairing $\varphi(\eta \vartriangleleft \xi)$ with an arbitrary element $\tilde{\nu}$ in $\mathfrak{h}^*_2$. We have that 
\begin{equation} \label{calc--33}
\begin{split}
\langle \varphi(\eta \vartriangleleft \xi),\tilde{\nu}\rangle
&=
 \langle \varphi(\eta) \vartriangleleft \varphi(\xi),\tilde{\nu}\rangle 
=
 \langle \varphi(\eta),\varphi(\xi)\overset{\ast }{\vartriangleright} \tilde{\nu}\rangle 
 \\ &= 
 \langle \eta,\varphi^*\big(\varphi(\xi)\overset{\ast }{\vartriangleright} \tilde{\nu}\big)\rangle ,
 \end{split}
\end{equation}
where identity (\ref{Lie-hom-eq}) is used in the first equality, and the definition of $\overset{\ast }{\vartriangleright}$ in \eqref{xi-star} is employed in the second equality. Further, 
\begin{equation} \label{calc--44}
\langle \varphi(\eta \vartriangleleft \xi),\tilde{\nu}\rangle
= \langle \eta \vartriangleleft \xi,\varphi^*(\tilde{\nu})\rangle
=\langle \eta, \xi \overset{\ast }{\vartriangleright} \varphi^*(\tilde{\nu})\rangle.
\end{equation}
\eqref{calc--33} and \eqref{calc--44} for an arbitrary $\tilde{\nu}$ manifest the second identity in \eqref{dual-triangle-eta-xi}.
\end{proof}

Consider a Lie algebra homomorphism $\varphi$ from a Lie algebra $\mathfrak{K}_1$ to  $\mathfrak{K}_2$. The commutation rule between the dual mapping $\varphi^*$ and the coadjoint action $\ad^*$ is computed to be 
\begin{equation} \label{dual-coad}
\ad^*_x \circ ~\varphi ^*=\varphi^*\circ\ad^*_{\varphi(x)}
\end{equation}    
for all $x$ in $\mathfrak{K}_1$. Notice that the coadjoint action on the left hand side is the one on $\mathfrak{K}^*_2$ whereas the coadjoint action on the right hand side is the one on $\mathfrak{K}^*_1$. This reads that $\varphi^*$ is a Poisson mapping preserving the Lie-Poisson brackets and the coadjoint flows.  
Now we are ready to study matched Lie-Poisson dynamics under Lie algebra homomorphisms preserving the matched pair decompositions. 
\begin{proposition} \label{phi-prop}
Assume that $\varphi$ is a Lie algebra homomorphism from a matched pair Lie algebra $\mathfrak{K}_1=\mathfrak{g}_1 \bowtie \mathfrak{h}_1$ to a matched pair Lie algebra $\mathfrak{K}_2=\mathfrak{g}_2 \bowtie \mathfrak{h}_2$ respecting decompositions. Then we have the following commutation law for the coadjoint actions and the pull-back $\varphi^*$ 
\begin{equation} \label{coad-phi*}
 \ad_{\xi\oplus\eta}^{\ast} (\varphi^*\tilde{\mu}\oplus\varphi^*\tilde{\nu})=  \varphi^*\circ\ad_{\varphi(\xi)\oplus\varphi(\eta)}^{\ast} (\tilde{\mu}\oplus\tilde{\nu})
\end{equation}
for any $\xi\oplus \eta $ in $\mathfrak{g}_1\bowtie \mathfrak{h}_1$ and for any $\tilde{\mu}\oplus \tilde{\nu} $ in $\mathfrak{g}_2\bowtie \mathfrak{h}_2$. 
\end{proposition}
\begin{proof}
One way to prove the proposition above is to apply directly the inclusions in \eqref{dual-inc} to the identity  \eqref{dual-coad}. However, we prefer once more a bit  longer proof in order to identify how the terms in the coadjoint actions behave under the Poisson mapping. For this end, we start with the left hand side of (\ref{coad-phi*}) for the explicit expression of the coadjoint representation in \eqref{ad-*}. Accordingly, we compute
\begin{equation}
\begin{split}
 &\ad_{\xi\oplus\eta}^{\ast} (\varphi^*\tilde{\mu}\oplus\varphi^*\tilde{\nu})
 \\ & \quad =\big(ad^{\ast}_{\xi}\circ  \varphi^*\tilde{\mu} 
 -(\varphi^*\tilde{\mu}) \overset{\ast }{	\vartriangleleft}\eta - \mathfrak{a}_{\eta}^{\ast}\circ \varphi^*\tilde{\nu}\big) \\ &\hspace{3cm}\oplus \big(
ad^{\ast}_{\eta}\circ \varphi^*\tilde{\nu} +\xi \overset{\ast }{\vartriangleright}(\varphi^*\tilde{\nu})+ \mathfrak{b}_{\xi}^{\ast}\circ\varphi^*\tilde{\mu}\big )
\\
& \quad =\big(  \varphi^*\circ ad^{\ast}_{\varphi(\xi)}\tilde{\mu} 
 -\varphi^*\big(\tilde{\mu} \overset{\ast }{\vartriangleleft} \varphi(\eta)\big)
  - \varphi^*\circ \mathfrak{a}^*_{\varphi(\eta)} \tilde{\nu}\big)
 \\& \hspace{3cm} \oplus
   \big(\varphi^* \circ 
ad^{\ast}_{\varphi(\eta)} \tilde{\nu} 
+
\varphi^*(\varphi(\xi)\overset{\ast }{\vartriangleright} \tilde{\nu})
+ 
\varphi^*\circ \mathfrak{b}^*_{\varphi(\xi)} \tilde{\mu}\big )
\\
& \quad = \varphi^*
\big(  ad^{\ast}_{\varphi(\xi)}\tilde{\mu} 
 - \big(\tilde{\mu} \overset{\ast }{\vartriangleleft} \varphi(\eta)\big)
  -  \mathfrak{a}^*_{\varphi(\eta)} \tilde{\nu}\big)
  \\& \hspace{3cm} \oplus
   \big( 
ad^{\ast}_{\varphi(\eta)} \tilde{\nu} 
+
 (\varphi(\xi)\overset{\ast }{\vartriangleright} \tilde{\nu})
+ 
 \mathfrak{b}^*_{\varphi(\xi)} \tilde{\mu}\big )
 \\
 & \quad
 =\varphi^*\circ\ad_{\varphi(\xi)\oplus\varphi(\eta)}^{\ast} (\tilde{\mu}\oplus\tilde{\nu}),
\end{split}
\end{equation}
where we have employed the identities presented in Lemmas \ref{b-a-*} and \ref{dual-b-a-} in the second line of the calculation, and we have used the linearity of the dual mapping in the third line. 
\end{proof}

\subsection{Euler-Poincar\'{e} Equations}~

Consider a Lagrangian function $\mathfrak{L}$ defined on a Lie algebra $\mathfrak{g}$. To arrive at the equations of motion governed by $\mathfrak{L}$, one takes the variation of the action functional 
\begin{equation}
\delta\int_{a}^{b}\mathfrak{L}\left( \xi\right) dt=\int_{a}^{b}
\left\langle \frac{\delta\mathfrak{L}}{\delta\xi},\delta\xi\right\rangle
_{e}
 dt. 
\end{equation}
Applying the reduced variational principle $\delta\xi=\dot{\eta}+\left[ \xi,\eta\right]$ 
to the Lie algebra element, one arrives at the Euler-Poincar\'{e} equations \cite{cendra2003variational,MarsdenRatiu-book}
\begin{equation}\label{EPEq}
\frac{d}{dt}\frac{\delta \mathfrak{L}}{\delta\xi}=-\ad_{\xi}^{\ast}\frac{\delta \mathfrak{L}}{%
	\delta\xi}.   
\end{equation}

Our aim in this section is to couple two different Euler-Poincar\'{e} dynamics described by \eqref{EPEq}. This can be done in several ways: the naive approach is to put the two equations together, a procedure that we call direct coupling. In such case, one has two Lie algebras, say $\mathfrak{g}$ and $\mathfrak{h}$, and assumes only trivial representations of each other. Physically, this corresponds to the situation when each of the systems has its own individual motion. Instead of two trivial actions, one can assume only a one-sided action, i.e., a non-trivial right action of $\mathfrak{g}$ on $\mathfrak{h}$. This corresponds algebraically with the Lie algebra structure  
\begin{equation}\label{mpla-right}
\lbrack (\xi\oplus\eta),\,(\tilde{\xi}\oplus\tilde{\eta})]= [\xi,\tilde{\xi}]\oplus \big(
[\eta,\tilde{\eta}]+\eta\vartriangleleft \tilde{\xi}-\tilde{\eta}\vartriangleleft \xi\big).  
\end{equation}    
on the semi-direct product Lie algebra $\mathfrak{g}\ltimes \mathfrak{h}$. It is immediate to see that the Lie algebra bracket \eqref{mpla-right} is a particular instance of the matched pair Lie algebra in \eqref{mpla}, where the left action in assumed to be trivial. In this realization, referring to Proposition \ref{ad-*-prop}, on the dual space $\mathfrak{g}^* \oplus \mathfrak{h}^*$, the coadjoint action is computed to be 
\begin{equation}  \label{ad-*-right}
\ad_{(\xi\oplus\eta)}^{\ast}(\mu\oplus\nu)=\underbrace{ \big(ad^{\ast}_{\xi} \mu - \mathfrak{a}_{\eta}^{\ast}\nu\big)}_{\in ~ \mathfrak{g}^*}\oplus \underbrace{ \big(
ad^{\ast}_{\eta} \nu +\xi \overset{\ast }{\vartriangleright}\nu\big )}_{\in ~ \mathfrak{h}^*}.
\end{equation}
Then assuming a Lagrangian function $\mathfrak{L}=\mathfrak{L}\left( \xi\oplus \eta\right)$ on the Lie algebra $\mathfrak{g}\ltimes \mathfrak{h}$ and identifying the variations
\begin{equation}
\frac{\delta \mathfrak{L}}{\delta(\xi\oplus \eta)}=\frac{\delta \mathfrak{L}}{\delta \xi}\oplus \frac{\delta \mathfrak{L}}{\delta \eta}\in \mathfrak{g}^* \oplus \mathfrak{h}^*
\end{equation}
one can write the Euler-Poincar\'{e} equations on the semi-direct product Lie algebra. Referring to \eqref{EPEq}, the dynamical equations are computed to be 
\begin{equation}\label{mEP-right}
\frac{d}{dt}\frac{\delta\mathfrak{L}}{\delta\xi}    =-\ad_{\xi}^{\ast}
\frac{\delta\mathfrak{L}}{\delta\xi}+\underbrace{\mathfrak{a}_{\eta}^{\ast}\frac
{\delta\mathfrak{L}}{\delta\eta}}_{\text{action of $\mathfrak{g}$}},\qquad 
\frac{d}{dt}\frac{\delta\mathfrak{L}}{\delta\eta}    =-\ad_{\eta}^{\ast}%
\frac{\delta\mathfrak{L}}{\delta\eta}-\underbrace{\xi\overset{\ast}{\vartriangleright
}\frac{\delta\mathfrak{L}}{\delta\eta}}_{\text{action of $\mathfrak{g}$}}. 
\end{equation}
The first terms on the right hand sides of the equations (\ref{mEP-1}) are the individual Euler-Poincar\'{e} motions on the Lie algebras $\mathfrak{g}$ and $\mathfrak{h}$, respectively. To see this, compare those terms with the Euler-Poincar\'{e} equation in (\ref{EPEq}).  
We have labelled the rest of the terms on the right hand side of the equations to exhibit the manifestations of the right action. We refer to a surely incomplete list \cite{cendra1998lagrangian,
	n2001lagrangian, holm1998euler, esen2014tulczyjew,MarsRatiWein84} for more details on Lagrangian dynamics on semidirect products. 
	
	Instead of a nontrivial right action, one can consider a nontrivial left action of $\mathfrak{h}$ on $\mathfrak{g}$ and arrive alternatively at the semidirect product Lie algebra $\mathfrak{g}\rtimes \mathfrak{h}$ with a Lie bracket 
	\begin{equation}\label{mpla-left}
\lbrack (\xi\oplus \eta),\,(\tilde{\xi}\oplus \tilde{\eta})]=\big( [\xi,\tilde{\xi}]+\eta\vartriangleright \tilde{\xi}-\tilde{\eta}\vartriangleright \xi
\big)\oplus [\eta,\tilde{\eta}] .  
\end{equation}
Once more, we have derived this bracket from the matched pair Lie algebra \eqref{mpla} by employing a trivial right action. In this case, the coadjoint action given in Proposition \ref{ad-*-prop} reduces to 
\begin{equation}  \label{ad-*-left}
\ad_{(\xi\oplus\eta)}^{\ast}(\mu\oplus\nu)=\underbrace{ \big(ad^{\ast}_{\xi} \mu -\mu \overset{\ast }{%
	\vartriangleleft}\eta \big)}_{\in ~ \mathfrak{g}^*}\oplus \underbrace{ \big(
ad^{\ast}_{\eta} \nu + \mathfrak{b}%
_{\xi}^{\ast}\mu\big )}_{\in ~ \mathfrak{h}^*}.
\end{equation}
Now, we are ready to recast the Euler-Poincar\'{e} equations on the semidirect product Lie algebras 
$\mathfrak{g}\rtimes \mathfrak{h}$ as 
\begin{equation}\label{mEP-1-left}
 \frac{d}{dt}\frac{\delta\mathfrak{L}}{\delta\xi}     =-\ad_{\xi}^{\ast}%
\frac{\delta\mathfrak{L}}{\delta\xi}+\underbrace{\frac{\delta\mathfrak{L}}{\delta\xi
}\overset{\ast}{\vartriangleleft}\eta}_{\text{action of $\mathfrak{h}$}}, \qquad 
\frac{d}{dt}\frac{\delta\mathfrak{L}}{\delta\eta}     =-\ad_{\eta}^{\ast}%
\frac{\delta\mathfrak{L}}{\delta\eta} -\underbrace{\mathfrak{b}_{\xi}^{\ast}\frac
{\delta\mathfrak{L}}{\delta\xi}}_{\text{action of $\mathfrak{h}$}}. 
\end{equation}
As a last comment on the semidirect theory, we recall the matched pair compatibility conditions in \eqref{compcon-mpl}. It is evident for the semidirect product Lie algebras that these conditions reduce to the simple relations determining the left or right character of the actions. 

It is easy to observe now that mutual actions are beyond the realm of the semidirect product theory. We think that matched pair Lie algebras are proper for studying mutually acting Euler-Poincar\'{e} flows. Let us depict this realization and discuss how matched pair Euler-Poincar\'{e} equations contain the semidirect product theories as a particular instance. For this general theory, we consider two Lie algebras, say $\mathfrak{h}$ and $\mathfrak{g}$, under mutual interaction as given in \eqref{Lieact} assuming the conditions in \eqref{compcon-mpl}. Then, it is immediate to observe that one can define a matched pair Lie algebra bracket \eqref{mpla} on the product space $\mathfrak{g}\bowtie \mathfrak{h}$. For a Lagrangian function(al) $\mathfrak{L}=\mathfrak{L}\left( \xi,\eta\right)$ depending on $\xi$ in $\mathfrak{g}$, and $\eta$ in $\mathfrak{h}$, the matched Euler-Poincar\'{e} equations are computed to be
\begin{equation}
\begin{split}\label{mEP-1}
\frac{d}{dt}\frac{\delta\mathfrak{L}}{\delta\xi}  &  =-\ad_{\xi}^{\ast}%
\frac{\delta\mathfrak{L}}{\delta\xi}+\underbrace{\frac{\delta\mathfrak{L}}{\delta\xi
}\overset{\ast}{\vartriangleleft}\eta}_{\text{action of $\mathfrak{h}$}}+\underbrace{\mathfrak{a}_{\eta}^{\ast}\frac
{\delta\mathfrak{L}}{\delta\eta}}_{\text{action of $\mathfrak{g}$}},\\
\frac{d}{dt}\frac{\delta\mathfrak{L}}{\delta\eta}  &  =-\ad_{\eta}^{\ast}%
\frac{\delta\mathfrak{L}}{\delta\eta}-\underbrace{\xi\overset{\ast}{\vartriangleright
}\frac{\delta\mathfrak{L}}{\delta\eta}}_{\text{action of $\mathfrak{g}$}}-\underbrace{\mathfrak{b}_{\xi}^{\ast}\frac
{\delta\mathfrak{L}}{\delta\xi}}_{\text{action of $\mathfrak{h}$}},
\end{split}
\end{equation}
where we have used the matched pair coadjoint action in Proposition \ref{ad-*-prop}.
In (\ref{mEP-1}), the second term on the right hand side of the first equation and the third term on the right hand side of the second equation are obtained by dualizing the left action of the Lie algebra $\mathfrak{h}$ on $\mathfrak{g}$. So that, if this action is trivial, that is, we have a semi-direct product Lie algebra $\mathfrak{g}\rtimes \mathfrak{h}$, then the Euler-Poincar\'{e} dynamics is the one in (\ref{mEP-right}) without these terms. On the other hand, the third term on the right hand side of the first line, and the second equation on the right hand side of the second line are manifestation of the action of $\mathfrak{g}$ on $\mathfrak{h}$. If the action is trivial in this case, that is, those terms are identically zero, then we arrive at the Euler-Poincar\'{e} equations on the semi-direct product Lie algebra $\mathfrak{g}\rtimes \mathfrak{h}$ and hence, the semidirect product Euler-Poincar\'{e} equations in \eqref{mEP-1-left}. In this regard, the matched pair Euler-Poincar\'{e} equations (\ref{mEP-1}) involve both of these semi-direct theories so that it permits mutual interactions. The first approach to study the Euler-Poincar\'{e} equations from the point of view of the matched pair theory is presented in \cite{EsSu16} and we refer to \cite{EsKuSu21} for the matching of higher order Euler-Poincar\'{e} equations. 

Let us examine now the behaviour of the Euler-Poincar\'{e} formalism under differentiable transformations. Assume a Lie algebra homomorphism $\varphi:\mathfrak{K}_1\to\mathfrak{K}_2$. Then, according to the identity \eqref{dual-coad}, the pull-back operation for a Lagrangian function $\mathfrak{L}=\mathfrak{L}(\varphi(x))$ defined on $\mathfrak{K}_2$ reads 
\begin{equation} \label{dual-coad---}
\begin{split}
\varphi^*\frac{d}{dt}\left(\frac{\delta  \mathfrak{L}}{\delta \varphi(x)}\right) &= - \varphi^*\circ \ad^*_{\varphi(x)}\left(\frac{\delta  \mathfrak{L}}{\delta \varphi(x)}\right)
=-  
\ad^*_{x} \circ\, \varphi^* \left(\frac{\delta  \mathfrak{L}}{\delta \varphi(x)}\right)
\\&=\frac{d}{dt}\varphi^* \left(\frac{\delta  \mathfrak{L}}{\delta \varphi(x)}\right).
\end{split}
\end{equation}    
We conclude that a Lie algebra homomorphism respects the Euler--Poincar\'{e} flows. Furthermore, Proposition \ref{phi-prop} verifies that, if both of the domain and the image space admit matched pair decompositions in the realm of the equation \eqref{coad-phi*}, then the coadjoint actions are properly conserved.

\section{Symmetric Tensor Spaces} 

In this section, we summarize the notation and main definitions in the spaces of symmetric contravariant and covariant tensors.

\subsection{Symmetric Contravariant Tensors} \label{sec-symcovtenfields}~

We denote the space of $k$-th order symmetric
contravariant tensor fields on a manifold $\mathcal{Q}$ by $\mathfrak{T}^{\textbf{k}}\mathcal{Q}$. The space of zeroth order tensors $\mathfrak{T}^{\textbf{0}}\mathcal{Q}$ is the space $\mathcal{F}(\mathcal{Q})$ whereas the  first order tensors $\mathfrak{T}^{\textbf{1}}\mathcal{Q}$ are precisely smooth vector fields  
$\mathfrak{X}(\mathcal{Q})$. We take the sum of all orders  to define the space of symmetric
contravariant tensor fields
\begin{equation}
{\mathfrak{T}\mathcal{Q}}:=\sum _{\textbf{k}=0}^{\infty }\mathfrak{T}
^{\textbf{k}}\mathcal{Q}.
\end{equation}
A bold super script $\textbf{k}$ in the notation ${\mathfrak{T}\mathcal{Q}}$ stands to denote the $k$-th order symmetric
contravariant tensor fields. We reserve the bold notation to distinguish the space with indices. Accordingly, on a local chart $(q^i)$ over $\mathcal{Q}$, an
element of $\mathfrak{T}\mathcal{Q}$ is written as
\begin{equation} \label{X^n}
\mathbb{X}=\sum_{\textbf{k}=0}^\infty\mathbb{X}^{\textbf{k}}=\sum_{k=0}^{\infty }\mathbb{X}^{i_{1}i_{2}...i_{k}}(q)
\partial {q^{i_{1}}}\otimes ...\otimes \partial {q^{i_{k}}},
\end{equation}
where $\mathbb{X}^{i_{1}i_{2}...i_{k}}$ are functions on $Q$. 

For a $k$-th order symmetric contravariant tensor field $\mathbb{X}^{\textbf{k}}$ and an $m$-th order field $\mathbb{Y}^{\textbf{m}}$, where $k+m\geq 1$, the symmetric Schouten concomitant is defined to be \cite{KoMiSl93,Ma97,Sc40}
\begin{equation} \label{SC-def}
\begin{split}
\left[ \mathbb{X}^{\textbf{k}},\mathbb{Y}^{\textbf{m}}\right]&:=\left(k\mathbb{X}^{i_{m+1}...i_{m+k-1}\ell}\mathbb{Y}^{i_{1}...i_{m}}_{,\ell}
-
m\mathbb{Y}^{i_{k+1}...i_{k+m-1}\ell}  \mathbb{X}^{i_{1}i_{2}...i_{k}}_{,\ell}\right)\\&\hspace{3cm} \partial {q^{i_{1}}}\otimes ...\otimes \partial {q^{i_{k+m-1}}}.
\end{split}
\end{equation}
See that, the result is a symmetric contravariant tensor field of order $m+k-1$. Assuming $\left[\mathbb{X}^\textbf{0},\mathbb{Y}^\textbf{0}\right] = 0$, we define a Lie algebra structure on $\mathfrak{T}\mathcal{Q}$ as follows. For $\mathbb{X}=\sum_{\mathbf{k\geq 0}}\mathbb{X}^\mathbf{k}$ and $\mathbb{Y}=\sum_{\mathbf{m\geq 0}}\mathbb{Y}^\mathbf{m}$, the symmetric Schouten concomitant  is 
\begin{equation}\label{sc}
\left[ \mathbb{X},\mathbb{Y}\right]_S=\sum_{\textbf{k,m}=0}^{\infty
}\left[ \mathbb{X}^{\textbf{k}},\mathbb{Y}^{\textbf{m}}\right]_S.
\end{equation}
If $\mathbb{X}^{\textbf{k}}$ is  a first order tensor field $X$ in (\ref{SC-def}), then the bracket reduces to the Lie derivative
\begin{equation} \label{Lie-multi}
\left[ X,\mathbb{Y}^{\textbf{m}}\right]=\mathcal{L}_X \mathbb{Y}^{\textbf{m}}= (X^{\ell} \mathbb{Y}^{i_{1}...i_{m}}_{,\ell} - m\mathbb{Y}%
^{i_{2}...i_{m}\ell} X^{i_{1}}_{,\ell}) \partial {q^{i_{1}}}\otimes ...\otimes \partial {q^{i_{m}}}
\end{equation} 
of the tensor field $\mathbb{Y}^{m}$ in the direction of $X$. We introduce the divergence 
\begin{equation} \label{diver}
\begin{split}
{\rm div} &: \mathfrak{T}^{\textbf{k}}\mathcal{Q} \longrightarrow \mathfrak{T}^{\textbf{k-1}}\mathcal{Q}, \\ &\qquad \mathbb{X}^{i_{1}i_{2}\dots i_{k}}
\partial {q^{i_{1}}}\otimes \dots \otimes \partial {q^{i_{k}}} \mapsto k\mathbb{X}^{\ell i_{2}\dots i_{k}}_{,\ell}
\partial {q^{i_{2}}}\otimes \dots \otimes \partial {q^{i_{k}}},
\end{split}
\end{equation}
for $k>0$. Notice that, if $k=1$, then the operation in \eqref{diver} turns out to be the classical divergence of a vector field. We define ${\rm div} \mathbb{X}^{0}$ as $0$.  

Let us point out two Lie subalgebras of  $\mathfrak{T}\mathcal{Q}$. The first one is 
\begin{equation} \label{s}
\mathfrak{s}:=\sum _{\textbf{k}=0}^{1}\mathfrak{T} ^{\textbf{k}}\mathcal{Q}=\mathcal{F}(\mathcal{Q})  \rtimes \mathfrak{X}(\mathcal{Q}).
\end{equation}
In this case, the symmetric Schouten concomitant \eqref{sc} reduces to the semi-direct product Lie algebra structure
\begin{equation}\label{pa}
\left[ \left(\rho,Z \right) ,\left(\sigma,Y \right) \right]
=\left(Z(\sigma) -Y(\rho), [ Z,Y]_{JL} \right),
\end{equation}
for $\left(\rho,Z \right)$ and $\left(\sigma,Y \right)$ in $\mathfrak{s}$. Notice that the complementary subspace of $\mathfrak{s}$ is determined by  
\begin{equation} \label{n}
\mathfrak{n}:=\sum_{\textbf{k}=2}^\infty\,\mathfrak{T}
^{\textbf{k}}\mathcal{Q}
\end{equation} 
and it is a subalgebra of $\mathfrak{T}\mathcal{Q}$ as well. Therefore, we have two complementary Lie subalgebras of  $\mathfrak{T}\mathcal{Q}$. Due to the universal character of the matched pair decomposition, and in the light of Proposition \ref{universal-prop}, $\mathfrak{T}\mathcal{Q}$ is a matched pair Lie algebra. We record this fact in the following statement and we refer the reader to \cite{EsSu21} for the proof.  
\begin{proposition} \label{mpdTQ}
The pair of Lie subalgebras $\mathfrak{s}$ and $\mathfrak{n}$ exhibited in (\ref{s}) and (\ref{n}) of the space ${\mathfrak{T}\mathcal{Q}}$ of symmetric contravariant tensor fields is a matched pair of Lie algebras, and
\begin{equation} \label{mpTQ}
{\mathfrak{T}\mathcal{Q}} =\mathfrak{s}\bowtie \mathfrak{n}, \qquad \mathbb{X}=(\sigma,Y)\bowtie \mathbf{X}
\end{equation}
where $(\sigma,Y)$ is an element of $\mathfrak{s}$ whereas $\mathbf{X} = \sum_{\textbf{k}=2}^{\infty }\mathbb{X}^{\textbf{k}}$ is in $\mathfrak{n}$. Mutual actions are computed to be
\begin{equation} \label{actions} 
\begin{split}
 \vartriangleright&:\mathfrak{n}\otimes \mathfrak{s}\mapsto \mathfrak{s}, \qquad \mathbf{X}\vartriangleright(\sigma,Y)=(0,[\mathbb{X}^{2},\sigma]),\\
  \vartriangleleft&: \mathfrak{n}\otimes \mathfrak{s}\mapsto \mathfrak{n}, \qquad  \mathbf{X}\vartriangleleft(\sigma,Y)=\sum_{\textbf{k}=2}^{\infty }([\mathbb{X}^{\textbf{k+1}},\sigma]-\mathcal{L}_Y \mathbb{X}^{\textbf{k}}).
 \end{split}
\end{equation}
\end{proposition}

\subsection{Symmetric Covariant Tensors}~ 
\label{notations}

 We shall now consider the space $\mathfrak{T}_{\textbf{k}}^{\ast }\mathcal{Q}$ of symmetric covariant (compactly supported) tensor fields of order $k$ as the dual of the space of symmetric contravariant tensor fields $\mathfrak{T}_{\textbf{k}}\mathcal{Q}$  \cite{GiHoTr08}. We denote the sum of all dual spaces by
\begin{equation*}
\mathfrak{T}^{\ast }\mathcal{Q} := \sum_{\textbf{k}=0}^{\infty }\mathfrak{T}_{\textbf{k}}^{\ast }\mathcal{Q}.
\end{equation*}
In local coordinates $(q^i)$ on $\mathcal{Q}$, an element of $\mathfrak{T}^{\ast }\mathcal{Q}$ can be written as
\begin{equation*}
\mathbb{A}=\bigoplus_{\textbf{m}=0}^{\infty }\mathbb{A}_\textbf{m}=\bigoplus
_{m=0}^{\infty }\mathbb{A}_{i_1\ldots i_m}( q) dq^{i_1}\odots dq^{i_m}.
\end{equation*}
Given the local tensor fields
\begin{equation*}
\mathbb{A}_{\textbf{m}}=\mathbb{A}_{i_{1}\ldots i_{m}} dq^{i_1}\odots dq^{i_m} \in \mathfrak{T}^{\ast}_\textbf{m}\mathcal{Q}, \quad \mathbb{X}^{\textbf{k}} = \mathbb{X}^{i_{1} ...i_{k}} \p{q^{i_1}}\odots \p{q^{i_k}} \in \mathfrak{T}^\textbf{k}\mathcal{Q},
\end{equation*} 
we denote the tensor contraction operation by $\lrcorner$, and
\begin{equation}
\mathbb{X}^{\textbf{k}}\lrcorner\mathbb{A}_{\textbf{m}} := \begin{cases}
\displaystyle \mathbb{X}^{i_{1}i_{2}...i_{k}}\mathbb{A}_{i_1i_{2}\ldots i_m} dq^{i_{k+1}}\odots dq^{i_m} \in \mathfrak{T}^\ast_{\textbf{m-k}}\mathcal{Q} & \text{if  } m> k, \\[.3cm]
\displaystyle \mathbb{X}^{i_{1}i_{2}...i_{k}}\mathbb{A}_{i_{1}i_{2}\ldots i_m}\p{q^{i_{m+1}}}\odots \p{q^{i_k}}  \in \mathfrak{T}^{\textbf{k-m}}\mathcal{Q} & \text{if  } k>m. 
\end{cases} \label{llcorner}
\end{equation} 
As an application and for future reference, using the operation \eqref{diver}, we compute the following contraction
\begin{equation} \label{div-cont}
{\rm div} \mathbb{X}^{\textbf{k}}  \lrcorner  \mathbb{A}_{\textbf{k+m-1}} = k \mathbb{X}^{\ell i_{m+1}\ldots i_{k+m-1}}_{,\ell} \mathbb{A}_{i_1\ldots i_{k+m-1}} dq^{i_1} \odots dq^{i_m}.
\end{equation}
After fixing a volume form $d\textbf{q}$ on $Q$, the duality between $\mathfrak{T}_{\textbf{k}}^{\ast }\mathcal{Q}$ and $\mathfrak{T}^\textbf{k}\mathcal{Q}$ is expressed in a  multiply-and-integrate form
\begin{equation}\label{pairing-1}
\left\langle \mathbb{A},\mathbb{X}\right\rangle =   \sum_{k\geq 0}\int_{\mathcal{Q}} \mathbb{A}_{i_1 \ldots i_k} (q) \mathbb{X}^{i_1\ldots i_k} (q) d\textbf{q}.
\end{equation}
We assume the topological conditions which make the pairing \eqref{pairing-1} convergent and we make use of the abbreviations given in  \cite{EsSu21}
\begin{equation}\label{abbri}
\begin{split}
\mathbb{A}_{\textbf{m+k-1}}\star \mathbb{X}^{\textbf{k}}&=m \mathbb{A}_{i_1\ldots i_{m-1} i_{m+1}\dots i_{m+k}} \mathbb{X}^{i_{m+1} \dots i_{m+k}}_{,i_m} dq^{i_1} \odots dq^{i_m} \in \mathfrak{T}^\ast_{\textbf{m}}\mathcal{Q}
\\
\mathbb{X}^{\textbf{k}} \ast \mathbb{A}_{\textbf{m+k-1}}&= k\mathbb{X}^{i_{m+1} \dots i_{m+k-1} \ell} \mathbb{A}_{i_1 \dots   i_{m+k-1}, \ell} dq^{i_1} \odots dq^{i_m}  \in \mathfrak{T}^\ast_{\textbf{m}}\mathcal{Q},
\end{split}
\end{equation}
for $k\geq 0$ and $m+k-1\geq 0$. Then by simply adding these two operations we introduce
\begin{equation}\label{Lie-gen}
{\rm L}_{\mathbb{X}^{\textbf{k}}} \mathbb{A}_{\textbf{m+k-1}}= \mathbb{A}_{\textbf{m+k-1}}\star \mathbb{X}^{\textbf{k}}+ \mathbb{X}^{\textbf{k}}\ast \mathbb{A}_{\textbf{m+k-1}} \in \mathfrak{T}^\ast_{\textbf{m}}\mathcal{Q}.
\end{equation}
If $m=0$ then $\mathbb{A}_{\textbf{k-1}}\star \mathbb{X}^{\textbf{k}}$ identically vanishes for all $k$, so that
\begin{equation} \label{gen-Lie-k}
{\rm L}_{\mathbb{X}^{\textbf{k}}} \mathbb{A}_{\textbf{k-1}}=\mathbb{A}_{\textbf{k-1}}\star \mathbb{X}^{\textbf{k}}+ \mathbb{X}^{\textbf{k}}\ast \mathbb{A}_{\textbf{k-1}}= \mathbb{X}^{\textbf{k}}\ast \mathbb{A}_{\textbf{k-1}} \in \mathfrak{T}^\ast_{0}\mathcal{Q}.
\end{equation}
In this case, $k$ must be greater than $0$. If, on the other hand, $k=0$ (that is $\mathbb{X}^{0}=\mathfrak{s}$ is a smooth function) in \eqref{Lie-gen}, then $\s \ast \mathbb{A}_{\textbf{m-1}}$  identically vanishes for all $m$, so that 
\begin{equation} \label{0-k}
\begin{split}
{\rm L}_{\s} \mathbb{A}_{\textbf{m-1}}&=\mathbb{A}_{\textbf{m-1}}\star \s + \s \ast \mathbb{A}_{\textbf{m-1}}=\mathbb{A}_{\textbf{m-1}}\star \s
\\&=m \mathbb{A}_{i_1\ldots i_{m-1}} \sigma_{,i_m} dq^{i_1} \odots dq^{i_m} \in \mathfrak{T}^\ast_{\textbf{m}}\mathcal{Q},
\end{split}
\end{equation}
provided that $m>0$. If, further $m=1$, then
\begin{equation} \label{0-0}
{\rm L}_{\mathbb{X}^{0}} \mathbb{A}_{0}={\rm L}_\sigma \rho = \sigma \ast  \rho+ \rho \star \sigma= 0 + \rho \star \sigma= \rho d \sigma.
\end{equation}
Finally, if $k=1$ then the notation reduces to the classical Lie derivative definition
\begin{equation} \label{1-m}
\begin{split}
{\rm L}_{\mathbb{X}^{1}} \mathbb{A}_{\textbf{m}}&=\mathcal{L}_Y\mathbb{A}_{\textbf{m}} \\&= \left ( m \mathbb{A}_{i_1\ldots i_{m-1} \ell} Y^{\ell}_{i_m} +  Y^{\ell} \mathbb{A}_{i_1 \dots   i_{m},\ell} \right) dq^{i_1} \odots dq^{i_m}  \in \mathfrak{T}^\ast_{\textbf{m}}\mathcal{Q}.
\end{split}
\end{equation}

The dual spaces of the Lie subalgebras $\mathfrak{s}$ in (\ref{s}) and $\mathfrak{n}$ in (\ref{n}) are
\begin{align}
\mathfrak{s}^{\ast} :& = \bigoplus_{\textbf{m}= 0}^1\mathbb{A}_{\textbf{m}}=\B{A}_0 \oplus \B{A}_1 \label{s*}
\\ \label{n*}
\mathfrak{n}^{\ast} :& = \bigoplus_{\textbf{k}\geq 2}\mathbb{A}_{\textbf{k}}=\B{A}_2 \oplus \B{A}_3 \oplus \dots
\end{align} 
respectively. Therefore, we arrive at the decomposition 
\begin{equation} \label{decomp-covariant-dual}
\mathfrak{T}^{\ast }\mathcal{Q}= \mathfrak{s}^{\ast}  \oplus  \mathfrak{n}^{\ast}, \qquad \mathbb{A}=(\rho,M) \oplus \mathbf A,
\end{equation}
where $(\rho,M)$ in $\mathfrak{s}^*$, whereas $\mathbf{A}=\sum_{k=2}^{\infty }\mathbb{A}_{\textbf{k}}$ in $\mathfrak{n}^\ast$. 

\section{Diffeomorphims Groups and Hamiltonian Vector Fields}

Let $\mathcal{M}$ be a smooth volume manifold. The
group of diffeomorphisms, denoted by ${\rm Diff} (\mathcal{M})$, on $\mathcal{M}$
is an infinite dimensional Lie group with multiplication 
\begin{equation}
{\rm Diff}(\mathcal{M}) \times {\rm Diff} ( \mathcal{M} )
\rightarrow {\rm Diff}(\mathcal{M}), \qquad ( \varphi ,\psi  )
\rightarrow \varphi \circ \psi
\end{equation}%
and inversion $\varphi \rightarrow \varphi ^{-1}.$ The unit element of the
group is the identity automorphism ${\rm id}$. As a manifold, 
${\rm Diff}(\mathcal{M}) $ is locally diffeomorphic to an infinite 
dimensional vector space, which can be a Banach, Hilbert or Fr\'{e}chet
space, and called respectively Banach Lie group, Hilbert Lie group or Fr\'{e}%
chet Lie Group \cite{sh04}. We will not discuss the details of the
functional analytical issues and refer to \cite{ChMa,em70}.

The elements of the tangent space $T_{\varphi }{\rm Diff}(\mathcal{M}) 
$ at $\varphi$  are material
velocity fields
\begin{equation}
V_{\varphi }:\mathcal{M}\rightarrow T\mathcal{M},
\end{equation}%
satisfying $\tau _{{\rm Diff}(\mathcal{M}) }\circ V_{\varphi }=\varphi 
$. In particular, the tangent space at the identity $T_{\rm id}{\rm Diff} ( \mathcal{M} ) $ is the space of smooth vector fields on $\mathcal{M}$,
that is, 
\begin{equation}
T_{\rm id}{\rm Diff}(\mathcal{M}) =\mathfrak{X} ( \mathcal{M}%
 ) .
\end{equation}%
A vector field on ${\rm Diff} ( \mathcal{M} ) $ is a map $V$ on ${\rm Diff}( \mathcal{M} ) $  taking values on the tangent bundle $T {\rm Diff}( \mathcal{M} ) $. A particular value of a vector field
at $\varphi$ in $ {\rm Diff} (\mathcal{M}) $ is the material velocity
field $V_{\varphi }$ in $T_{\varphi } {\rm Diff}(\mathcal{M}) $. $%
V_{\varphi }$ can be represented as a composition of a diffeomorphism $%
\varphi $ and a vector field $X$, that is 
\begin{equation}
V_{\varphi }=X\circ \varphi.
\end{equation}
This is the manifestation of the parallelizability of $ T{\rm Diff} ( \mathcal{M} )$.

We assume that a continuum rests in $\mathcal{M}$ and ${\rm Diff} (\mathcal{M}) $ acts on left by evaluation on the space $\mathcal{M}$ that is
\begin{equation}
{\rm Diff}(\mathcal{M}) \times \mathcal{M}\rightarrow \mathcal{M}%
, \qquad (\varphi ,\mathbf{x})\rightarrow \varphi \left( \mathbf{x}\right)
\end{equation}%
to reproduce the motion of particles. The right action of ${\rm Diff} ( \mathcal{M} ) $ commutes with the particle motion and constitutes an infinite
dimensional symmetry group of the kinematical description. This is known as particle relabelling symmetry \cite{arkh}.

The inner automorphism on the group ${\rm Diff}(\mathcal{M}) $ is 
\begin{equation}
I_{\psi }\left( \varphi ^{t}\right) =\psi \circ \varphi ^{t}\circ \psi ^{-1},
\end{equation}%
and its differentiation at $t=0$ along the direction $X$ gives adjoint
operator, that is 
\begin{eqnarray}
Ad_{\psi } ( X ) &=&T_{e}I_{\psi } ( X ) =T_{e}I_{\psi
}\big( \frac{d}{dt}\left. \varphi ^{t}\right\vert_{t=0}\big) =\frac{d}{dt%
}\left. I_{\psi }\varphi ^{t}\right\vert _{t=0}  \notag \\
&=&\frac{d}{dt}\left. \psi \circ \varphi ^{t}\circ \psi ^{-1}\right\vert
_{t=0}=T\psi \circ X\circ \psi ^{-1}=\psi _{\ast }X.
\end{eqnarray}%
Thus, the adjoint action of ${\rm Diff}(\mathcal{M})$ on its Lie algebra $\mathfrak{X}(\mathcal{M}) $ is the push-forward operation 
\begin{equation}
Ad_{\psi }\left( X\right) =\psi _{\ast }X.
\end{equation}%
The tangent space of ${\rm Diff}(\mathcal{M}) $ at the identity ${%
\rm id}$ consists of vector fields on $\mathcal{M}$. The Lie
algebra bracket on $T_{\rm  id}{\rm Diff}(\mathcal{M}) $ can be
calculated as the differential of the adjoint representation at the
identity. We differentiate $Ad_{\psi ^{t}} ( X ) $ with respect to $%
t$ at $t=0$ and in the direction of $Y$ to obtain 
\begin{equation}
\left[ Y,X\right] _{{\rm Diff}(\mathcal{M}) }=ad_{Y}X=\frac{d}{dt}%
\left. \psi _{\ast }^{t}X\right\vert _{t=0}=-\left[ Y,X\right] _{JL}=-%
\mathcal{L}_{Y}X,  \label{Liealgebra}
\end{equation}%
where $\left[ \text{ },\text{ }\right] _{JL}$ is the standard Jacobi-Lie
bracket of vector fields and $\mathcal{L}_{Y}X$ is the Lie derivative of $X$
with respect to $Y$. Thus, the Lie algebra structure is minus the Jacobi-Lie bracket.

The dual space of the Lie algebra $\mathfrak{X}(\mathcal{M})$ is the space of one-forms
densities on $\mathcal{M}$, that is, 
\begin{equation}
\mathfrak{X}^{\ast
}(\mathcal{M}) \simeq \Lambda ^{1}\left( \mathcal{M}\right)
\otimes Den(\mathcal{M}) .
\end{equation}%
The pairing is 
\begin{equation}
 \langle \alpha \otimes \varpi ,X \rangle =\int_{\mathcal{M}} \langle \alpha   ,X 
 \rangle \varpi   ,  \label{pairing-2}
\end{equation}%
where $X\in \mathfrak{X}(\mathcal{M}) \mathbf{,}$ $\alpha \in
\Lambda ^{1}(\mathcal{M}) $ and $\varpi $ is a volume form on $%
\mathcal{M}$. The pairing inside the integral is the natural pairing of
finite dimensional spaces $T_{x}\mathcal{M}$ and $T_{x}^{\ast }\mathcal{M}$. The dual $ad^{\ast }$ of the adjoint action $ad$ is
defined by%
\begin{equation}
\left\langle ad_{X}^{\ast }\left( \alpha \otimes \mu \right) ,Y\right\rangle
=-\left\langle \left( \alpha \otimes \mu \right) ,ad_{X}Y\right\rangle
=\int_{\mathcal{M}} \langle \alpha  ,\left[ X,Y
\right] _{JL}  \rangle \varpi
 ,
\end{equation}%
and after applying integration by parts, we find the explicit expression 
\begin{equation}
ad_{X}^{\ast }\left( \alpha \otimes \mu \right) =-\left( \mathcal{L}_{X}\alpha +\left( {\rm div}%
_{\varpi }X\right) \alpha \right) \otimes \varpi ,  \label{coadjdiff}
\end{equation}%
of the coadjoint action $ad^{\ast }$, where ${\rm div}_{\varpi}X$ is the
divergence of the vector field $X$ with respect to the volume form $\varpi$.
For the case of divergence free vector fields, (\ref{coadjdiff}) reduces to%
\begin{equation}
ad_{X}^{\ast }\alpha =-\mathcal{L}_{X}\alpha .  \label{newa}
\end{equation}

\subsection{Canonical Diffeomorphisms}

The group of canonical diffeomorphisms ${\rm Diff_{can}} ( T^{\ast }\mathcal{Q} ) $ on the canonical symplectic manifold $T^{\ast }\mathcal{Q}$ consists of diffeomorphisms $
\varphi $ preserving the symplectic form $\Omega
_{\mathcal{Q}}$, that is, $\varphi ^{\ast }\Omega _{\mathcal{Q}}=\Omega _{ \mathcal{Q}
}.$ This reads the conservation $\mathfrak{L}%
_{X}\Omega _{\mathcal{Q}}=0$ and the Cartan's formula $\mathfrak{L}%
_{X}=d\iota_{X}+\iota_{X}d$ leads to $d\iota_{X}\Omega _{\mathcal{Q}}=0.$

Therefore, the Lie algebra of ${\rm Diff_{can}} ( T^{\ast }\mathcal{Q} ) $ can be identified with the space of (locally)
Hamiltonian vector fields $\mathfrak{X}_{ham}\left(
T^{\ast }\mathcal{Q}\right) $, see \cite{leo01,Gu10,EsGu12}. The following equalities 
\begin{equation}\label{Hamaa}
\lbrack X_{h},X_{f}]_{\mathfrak{X}}=-[X_{h},X_{f}]_{JL}=X_{\{h,f\}}
\end{equation}%
set a correspondence between the space $\mathfrak{X}_{ham}\left( T^{\ast }\mathcal{Q}\right) $ and the space
of smooth functions $\mathcal{F}\left( T^{\ast }\mathcal{Q}\right)$. We write here the following Lie algebra homomorphism for future reference
\begin{equation}\label{epi-onto-Ham}
\varphi:\mathcal{F}(T^\ast\mathcal{Q}) \longrightarrow \mathfrak{X}_{\mathrm{\mathrm{ham}}}(T^\ast\mathcal{Q}), \qquad h \mapsto X_h.
\end{equation}
It is easy to see that the kernel of this mapping is the space of constant functions. 

The nonzero elements of the dual space of the Lie algebra $\mathfrak{X}_{ham}\left( T^{\ast }%
\mathcal{Q}\right) $ are given by
\begin{equation}\label{X*}
 \mathfrak{X}_{ham}^{\ast } ( T^{\ast }\mathcal{Q}%
 ) =\{\Pi \in \Lambda ^{1}(T^{\ast }\mathcal{Q}):{\rm div}
_{\Omega _{T^{\ast }Q}}\Pi  ^{\sharp }\neq 0\}.
\end{equation}

To find the precise definition of the dual space $\mathfrak{X}_{ham}^{\ast
}\left( T^{\ast }\mathcal{Q}\right) ,$ we require the $L_{2}$ pairing $%
\left\langle X_{h},\Pi  \right\rangle $ to be nondegenerate. We
take the volume $d\textbf{q}d\textbf{p}$ and compute 
\begin{equation}
\begin{split}
\int_{T^{\ast }\mathcal{Q}}\left\langle X_{h}  ,\Pi  \right\rangle d\textbf{q}d\textbf{p} 
&= -\int_{T^{\ast }\mathcal{Q}} \langle dh,\Pi  ^{\sharp }
\rangle d\textbf{q}d\textbf{p} 
=-\int_{T^{\ast }\mathcal{Q}}\iota_{\Pi  ^{\sharp }}\left( dh\right) d\textbf{q}d\textbf{p}  \notag \\
&=-\int_{T^{\ast }\mathcal{Q}}dh\wedge \iota_{\Pi ^{\sharp }}(d\textbf{q}d\textbf{p})
=\int_{T^{\ast }\mathcal{Q}}h~ d \iota_{\Pi ^{\sharp }}(d\textbf{q}d\textbf{p}) \\
&= \int_{T^{\ast }\mathcal{Q}}h~{\rm div}_{\Omega _{T^{\ast }Q}}\Pi  ^{\sharp }(d\textbf{q}d\textbf{p})  ,
\end{split}
\end{equation}
where we have used the musical isomorphism $\Omega _{\mathcal{Q}}^{\sharp
}$ induced from the
symplectic two-form $\Omega _{\mathcal{Q}}$ in the first step and we have applied integration by parts in the last step. Thus, in Darboux' coordinates,  we arrive at the following map
\begin{equation}
\Pi  \rightarrow f(q,p)={\rm div}_{\Omega
_{T^{\ast }Q}}\Pi  ^{\sharp }\left( \mathbf{z}\right), \qquad \Pi
_{i} dq^{i}+\Pi ^{i} dp_{i}\mapsto \frac{\partial \Pi ^{i}}{\partial q^{i}}-\frac{\partial \Pi _{i} }{\partial p_{i}}.
\label{mom1}
\end{equation}
which is defined to be the density function. Note that, if 
\begin{equation}
\Pi=\delta _{ij}\frac{\partial \psi }{\partial p_{i}}dq^{j}-\delta ^{ij}%
\frac{\partial \psi }{\partial q^{i}}dp_{j}
\end{equation}
for some function $\psi$, then
the identification in  (\ref{mom1}) reduces to the following Laplace
equation $f=\Delta \psi$.

It is important to remark that the action of ${\rm Diff_{can}} ( T^{\ast }\mathcal{Q} ) $ on $T^{\ast }\mathcal{Q}$ is a canonical action with momentum map 
\begin{equation}
\mathbf{J}:T^{\ast }\mathcal{Q}\rightarrow \mathfrak{X}_{ham}^{\ast } ( T^{\ast }\mathcal{Q}%
 ) , \qquad 
\left\langle \mathbf{J} ( \mathbf{z} ) ,X_{h}\right\rangle =h ( 
\mathbf{z} ) ,
\end{equation}%
where $X_{h}$ is the Hamiltonian vector field for the Hamiltonian function $h.$

\subsection{Generalized Complete Cotangent Lift}~

The lifting of a symmetric $k$-covariant tensor field  $\mathbb{X}^{\textbf{k}}$ on $\mathcal{Q}$ to a function on the cotangent bundle $T^*\mathcal{Q}$ is defined to be 
\begin{equation}\label{TkQ-to-F}
\kappa:{\mathfrak{T}^\textbf{k}\mathcal{Q}}\longrightarrow  \mathcal{F}(T^*\mathcal{Q}), \qquad \mathbb{X}^{\textbf{k}}  \mapsto \hat{\B{X}}^{\textbf{k}}:=\theta _{\mathcal{Q}}^{k} (\mathbb{X}^{\textbf{k}}),
\end{equation}
where $\theta_{ \mathcal{Q}%
}^{k}=\theta_{\mathcal{Q}}\otimes ...\otimes \theta _{\mathcal{Q}}$ is the $k$-th tensor power of the canonical-one form $\theta _{\mathcal{Q}}$ on $T^*\mathcal{Q}$. In the local picture of $\mathbb{X}^{\textbf{k}}$ in \eqref{X^n}, this operation reads a $p$-polynomial
\begin{equation*}
\hat{\B{X}}^{\textbf{k}} = \B{X}^{i_1\ldots i_k}(q)\,p_{i_1}\ldots p_{i_k}.
\end{equation*}
It is evident that the mapping \eqref{TkQ-to-F} is far from being surjective if one focuses on the smooth category. In this case, the space of flat functions with respect to the momentum variables cannot be obtained in the image space, whereas in the analytical category, the image space is equal to the functions on $T^*\mathcal{Q}$. Since flat functions cannot be observable for physical systems \cite{BlAs79}, we determine $\mathcal{F}(T^*\mathcal{Q})$ as the image space of the mapping $\kappa$. For a more general discussion involving the flat functions we refer once more to \cite{EsSu21}. 
Assuming the canonical Poisson bracket on $\mathcal{F}(T^*\mathcal{Q})$, we have that $\kappa$ is a Lie algebra anti-homomorphism, that is
\begin{equation}\label{sc-Poi}
\widehat{[\B{X},\B{Y}]}_S=-\{\hat{\B{X}},\hat{\B{Y}}\}
\end{equation}
where, referring to \eqref{TkQ-to-F}, we have that 
\begin{equation}\label{TQ-to-F}
\B{X}=\sum_{\textbf{k}=0}^\infty \B{X}^{\textbf{k}}\mapsto \hat{\B{X}} :=\sum_{\textbf{k}=0}^\infty \hat{\B{X}}^{\textbf{k}} .
\end{equation}

Composing $\kappa$ with the mapping \eqref{epi-onto-Ham} and multiplying by negative, we define the
generalized complete cotangent lift, abbreviated as {\small GCCL}, 
\begin{equation} \label{gccl}
\begin{split}
\text{{\small GCCL}} &:\varphi\circ \kappa=\mathfrak{T}\mathcal{Q}\longrightarrow \mathfrak{X}%
_{\mathrm{ham}}\left( T^{\ast }\mathcal{Q}\right), \\& \qquad \mathbb{X}=\sum_{\textbf{k}=0}^\infty \mathbb{X}^{\textbf{k}}\mapsto -X_{\hat{\mathbb{X}}}:=-\sum_{k=0}^\infty X_{\hat{\mathbb{X}}^{\textbf{k}}}
\end{split} 
\end{equation}
See that \eqref{gccl} takes a contravariant tensor field $\mathbb{X}$ on $\mathcal{Q}$
to minus the Hamiltonian vector field $X_{\hat{\mathbb{X}}}$
generated by the Hamiltonian function $\hat{\mathbb{X}}$ in the polynomial form \eqref{TkQ-to-F}. In the Darboux'
coordinates, {\small GCCL} is given by
\begin{equation}
\begin{split}
\text{{\small GCCL}} ( \mathbb{X}^{\textbf{k}} )
&=-X_{\hat{\mathbb{X}}^{\textbf{k}}} \\& =-kp_{i_{1}}p_{i_{2}}...p_{i_{k-1}}\mathbb{X}^{i_{1}...i_{k-1}\ell}\partial
_{q^{\ell}}+p_{i_{1}}p_{i_{2}}...p_{i_{k}}\frac{\partial \mathbb{X}%
^{i_{1}i_{2}...i_{k}}}{\partial q^{\ell}}\partial _{p_{\ell}}.  \label{Xnc}
\end{split} 
\end{equation}

\begin{lemma}
The generalized complete cotangent lift (\ref{gccl}) is a Lie algebra homomorphism, that is,
\begin{equation}
\text{{\small\rm GCCL}}  \left[ \mathbb{X},\mathbb{Y}\right]_S =\left[ \text{{\small \rm GCCL}} (\mathbb{X}),\text{{\small \rm GCCL}}( \mathbb{Y})\right]_{\mathfrak
{X}},  \label{iso}
\end{equation}%
where $\left[ \bullet,\bullet\right]_S $ is the Schouten concomitant (\ref{sc})
of tensor fields, and $\left[\bullet,\bullet\right]
_{\mathfrak{
X}}$ is the opposite Jacobi-Lie bracket of vector fields.
\end{lemma}

Let us provide the proof of this assertion. Accordingly, we compute
\begin{equation}
\begin{split}
\left[ \text{{\small \rm GCCL}} (\mathbb{X}),\text{{\small \rm GCCL}}( \mathbb{Y})\right]_{\mathfrak
{X}}&=-\left[ \text{{\small \rm GCCL}} (\mathbb{X}),\text{{\small \rm GCCL}}( \mathbb{Y})\right]_{JL}=- [-X_{\hat{\mathbb{X}}}, -X_{\hat{\mathbb{Y}}}]_{JL}\\ &=- [X_{\hat{\mathbb{X}}}, X_{\hat{\mathbb{Y}}}]_{JL}=X_{\{\hat{\mathbb{X}},\hat{\mathbb{Y}}\}}=-X_{\widehat{[\B{X},\B{Y}]}_S}\\
&=\text{{\small \rm GCCL}} [\B{X},\B{Y}]_S,  \label{iso-pf}
\end{split}
\end{equation}%
where we have employed the definition of {\small \rm GCCL} in the first line, identity \eqref{Hamaa} in the second equality displayed in the second line, and identity \eqref{sc-Poi}  in the second equality displayed in the third equality. 

We now exhibit the image of the constitutive Lie subalgebras $\mathfrak{s}$ and $\mathfrak{n}$ in the matched pair decomposition of $\mathfrak{T}\mathcal{Q}$ in Proposition \ref{mpdTQ} under {\small GCCL}. A direct calculation shows that the restriction of {\small GCCL} to $\mathfrak{s}$ reads
\begin{equation}\label{subem-s}
\begin{split} 
\mathfrak{s}\longrightarrow \mathfrak{s}^{c}&:=\text{{\small GCCL}}(\mathfrak{s}), \\&
(\sigma,Y) \mapsto -X_{\hat{\sigma}}- X_{\widehat{Y}} = \sigma_{,i} \partial
_{p_{i}}-X^{i}\partial _{q^{i}}+p_{j}%
X^{j}_{,i} \partial _{p_{i}},
\end{split} 
\end{equation}%
whereas  the restriction of {\small GCCL} to the Lie subalgebra $\mathfrak{n}$ is
\begin{equation}
\mathfrak{n}\longrightarrow \mathfrak{n}^{c}:=\text{{\small GCCL}}(\mathfrak{n}), \qquad \mathbf{X}=
\sum_{\textbf{k}=2}^\infty\mathbb{X}^{\textbf{k}} \mapsto -X_{\hat{\mathbf{X}}}:=-\sum_{\textbf{k}=2}^\infty X_{\hat{\B{X}}^\textbf{k}}. \label{subem-n}
\end{equation}
We refer to \cite{EsSu21} for proof of the following statement. 

\begin{proposition} \label{H_ham-decomp-Pro}
Let the pair of Lie algebras $(\mathfrak{s}^{c},\mathfrak{n}^{c})$  be the ones given by \eqref{subem-s}. Then, \eqref{subem-n} determines a  matched pair decomposition  of the space of Hamiltonian vector fields generated by non-flat smooth functions on $T^*\mathcal{Q}$ as 
\begin{equation} \label{H_ham-decomp-eq}
{\mathfrak{X}}_{\mathrm{ham}}(T^*\mathcal{Q}) \cong {\mathfrak{s}}^{c}\bowtie {\mathfrak{n}}^{c},
\end{equation} 
where ${\mathfrak{s}}^{c}$ is the image of $\C{F}(T^*\C{Q})/\mathbb{R}$ under the mapping $\varphi$ in \eqref{epi-onto-Ham}. 
\end{proposition}

Let us present the result of Proposition \ref{H_ham-decomp-Pro} by computing  the mutual actions for future reference. Notice that decomposition (\ref{H_ham-decomp-eq}) asserts that a Hamiltonian vector field $-X_h$ can be written as the pair
\begin{equation}
-X_h=X_h^{\mathfrak{s}} + X_h^{\mathfrak{n}}, 
\end{equation}
where the constitutive vector fields are defined to be 
\begin{equation}
X_h^{\mathfrak{s}}=\text{\small GCCL}(\sigma,Y)= -X_{\widehat{(\sigma,Y)}}, \qquad  X_h^{\mathfrak{n}}=\text{\small GCCL}(\textbf{X})=-X_{\hat
{\textbf{X}}}
\end{equation}
for some $(\sigma,Y)$ in $\G{s}$ and $\textbf{X}$ in $\G{n}$. In order to compute the mutual actions, we may refer to the universal property in Proposition \ref{universal-prop} once more, or instead, \eqref{Lie-hom-eq} given in Lemma \eqref{mp-homo}, so that we get 
\begin{equation}
\begin{split}
X_h^{\mathfrak{n}} \vartriangleright X_h^{\mathfrak{s}}
&=(
-X_{\hat{\mathbf{X}}}) \vartriangleright (-X_{\widehat{(\sigma,Y)}})=-X_{\widehat{\mathbf{X} \vartriangleright (\sigma,Y)  }}\\&=-X_{\widehat{[\mathbb{X}^2,\s]}}
,
\end{split}
\end{equation}
where we have used the left action in \eqref{actions}. In the light of the identity in \eqref{Lie-hom-eq}, the right action is computed to be
\begin{equation}
\begin{split}
X_h^{\mathfrak{n}} \vartriangleleft X_h^{\mathfrak{s}}
&=(
-X_{\hat{\mathbf{X}}}) \vartriangleleft (-X_{\widehat{(\sigma,Y)}})=-X_{\widehat{\mathbf{X} \vartriangleleft (\sigma,Y)  }}
\\&=\sum_{\textbf{k}=2}^{\infty }\Big( X_{\widehat{\mathcal{L}_Y \mathbb{X}^{\textbf{k}}}}-X_{\widehat{[\mathbb{X}^{\textbf{k+1}},\sigma]}}\Big).
\end{split}
\end{equation}

\section{Euler-Poincar\'{e} Flows}

In this section we introduce the decomposition of Euler-Poincar\'{e} dynamics on the space of contravariant tensor fields and on Hamiltonian vector fields.

\subsection{EP Dynamics on the Space Contravariant Tensor Fields}

Being a Lie algebra, the space $\mathfrak{T}\mathcal{Q}$
of contravariant fields determines coadjoint action on the space $\mathfrak{T}\mathcal{Q}$
of covariant fields $\mathfrak{T}\mathcal{Q}$ as follows
\begin{equation} \label{coad-TQ}
\langle \ad^\ast_{\B{X}}\B{A},\B{Y}\rangle =\langle
\B{A},[\B{Y},\B{X}]_S \rangle
\end{equation} 
for all $\B{X}$ and $\B{Y}$ in $\mathfrak{T}\mathcal{Q}$ and $\B{A}$ in $\mathfrak{T}^*\mathcal{Q}$. Here, the pairing is the one in \eqref{pairing-2} whereas the bracket is the Schouten concomitant \eqref{sc}. To have the explicit expression of the coadjoint action, we consider $\mathbb{X}^{\textbf{k}}, \mathbb{Y}^m$, and $ \mathbb{A}_{\textbf{m+k-1}}$, after fixing the volume form $d\textbf{q}$ on $\mathcal{Q}$, we perform  the following calculation
\begin{equation}
\begin{split}
&\big\langle  ad^\ast_{\mathbb{X}^{\textbf{k}}}\B{A}_{\textbf{m+k-1}}, \mathbb{Y}^\textbf{m} \big\rangle  =  \big\langle {\B{A}}_{\textbf{m+k-1}}, [\mathbb{Y}^\textbf{m}, \mathbb{X}^{\textbf{k}}] \big\rangle  \\
&=
 \int_{\mathcal{Q}} \mathbb{A}_{i_1\ldots i_{k+m-1}}\big(m\mathbb{Y}^{i_{k+1}\ldots i_{k+m-1}\ell} \mathbb{X}^{i_1\ldots i_k}_{,\ell} - k\mathbb{X}^{i_{m+1}\ldots i_{k+m-1}\ell} \mathbb{Y}^{i_1\ldots i_m}_{,\ell} \big)d\textbf{q} 
 \\
 &=m\int_{\mathcal{Q}} \mathbb{A}_{i_1\ldots i_{k+m-1}}\mathbb{Y}^{i_{k+1}\ldots i_{k+m-1}\ell} \mathbb{X}^{i_1\ldots i_k}_{,\ell}d\textbf{q} 
 \\
& \qquad  +k \int_{\mathcal{Q}} \,\mathbb{A}_{i_1\ldots i_{k+m-1},\ell}\mathbb{X}^{i_{m+1}\ldots i_{k+m-1}\ell}\mathbb{Y}^{i_1\ldots i_m} d\textbf{q} \\& \qquad  +k \int_{\mathcal{Q}}  \mathbb{A}_{i_1\ldots i_{k+m-1}} \mathbb{X}^{i_{m+1}\ldots i_{k+m-1}\ell}_{,\ell}\mathbb{Y}^{i_1\ldots i_m}d\textbf{q} 
\\ &=
\big\langle \mathbb{A}_{\textbf{m+k-1}}\star \mathbb{X}^{\textbf{k}},\mathbb{Y}^\textbf{m} \big\rangle 
+
 \big\langle \mathbb{X}^{\textbf{k}} \ast \mathbb{A}_{\textbf{m+k-1}} ,\mathbb{Y}^\textbf{m} \big\rangle
+
 \big\langle {\rm div} \mathbb{X}^{\textbf{k}}  \lrcorner  \mathbb{A}_{\textbf{k+m-1}},\mathbb{Y}^\textbf{m} \big\rangle,
 \end{split}
\end{equation}
where, in the second line, we have employed the explicit expression of the Schouten concomitant in \eqref{SC-def} and in the last line, we have referred to the definitions of $\star$ and $\ast$ in \eqref{abbri}, and the contraction of the divergence in \eqref{div-cont}. Recalling the notation in \eqref{Lie-gen}, we can collect the first two terms in the last line of the calculation, providing the following realization of the coadjoint action
\begin{equation} \label{Coad-VF}
 \ad^\ast_{\mathbb{X}^{\textbf{k}}}\mathbb{A}_{\textbf{m+k-1}} = {\rm L}_{\mathbb{X}^{\textbf{k}}} \mathbb{A}_{\textbf{m+k-1}} + {\rm div} \mathbb{X}^{\textbf{k}}  \lrcorner  \mathbb{A}_{\textbf{k+m-1}}\in \mathfrak{T}^\ast_{\textbf{m}}\mathcal{Q}.
\end{equation}
This calculation excludes the case $m=0$. To have that, we use \eqref{gen-Lie-k}.

We collect all of the discussions into the following proposition \cite{EsSu21}.  
\begin{proposition}\label{prop-coad}
The coadjoint action of $\mathfrak{T}\mathcal{Q}$ on $\mathfrak{T}^{\ast }\mathcal{Q}$ is given by
\begin{equation}\label{Coad-VF-}
\ad^\ast_{\B{X}}\B{A} = \bigoplus_{\textbf{m}=0}^\infty \widetilde{\mathbb{A}}_m,
\end{equation}
where 
\begin{equation}
\begin{split}  \label{Coad-VF-eq}
\widetilde{\mathbb{A}}_0
&=
\sum_{\textbf{k} = 1}^\infty\, 
  \mathbb{X}^{\textbf{k}}\ast \mathbb{A}_{ \textbf{k-1}}
+ {\rm div} \mathbb{X}^{\textbf{k}}  \lrcorner  \mathbb{A}_{\textbf{k-1}}
\\
\widetilde{\mathbb{A}}_m&=\sum_{\textbf{k} = 0}^\infty\,  {\rm L}_{\mathbb{X}^{\textbf{k}}} \mathbb{A}_{\textbf{m+k-1}} + {\rm div} \mathbb{X}^{\textbf{k}}  \lrcorner  \mathbb{A}_{\textbf{k+m-1}}, \qquad m\geq 1.
 \end{split}
\end{equation}
\end{proposition}

Now, we assume a Lagrangian function $\mathfrak{L}=\mathfrak{L}(\mathbb{X})$ on $\mathfrak{T}\mathcal{Q}$. 
The generic formulation in \eqref{EPEq} determines the Euler-Poincar\'{e} equations on the spce of contravariant tensor fields as 
 \begin{equation}\label{EP-comp-TQ}
\frac{d}{dt}\frac{\delta \mathfrak{L}}{\delta \mathbb{X}}=- \ad^\ast_{\B{X}}\frac{\delta \mathfrak{L}}{\delta \mathbb{X}}.
\end{equation}
We write an element of $\mathfrak{T}\mathcal{Q}$ as $(\sigma,Y,\sum_{\textbf{k}=2} \mathbb{X}^{\textbf{k}})$. We assume that the variation of the Lagrangian is 
\begin{equation}
\frac{\delta \mathfrak{L}}{\delta \mathbb{X}}= \frac{\delta \mathfrak{L}}{\delta \big( \sum_{\mathbf{m\geq 0}}\mathbb{X}^\mathbf{m}\big)}=\bigoplus_{\textbf{m}=0}  \frac{\delta \mathfrak{L}}{\delta \mathbb{X}^\mathbf{m}} = \big(\frac{\delta \mathfrak{L}}{\delta \sigma}\oplus \frac{\delta \mathfrak{L}}{\delta Y}\big)  \bigoplus_{\textbf{m}=2}  \frac{\delta \mathfrak{L}}{\delta \mathbb{X}^\mathbf{m}}
\end{equation}
 Then, in the light of Proposition \eqref{prop-coad}, we write the Euler-Poincar\'{e} equation explicitly as 
\begin{equation} \label{EP-TQ}
\begin{split}
\frac{d}{dt}\frac{\delta \mathfrak{L}}{\delta \sigma}&=
-\big(\mathcal{L}_Y(\frac{\delta \mathfrak{L}}{\delta \sigma})+ {\rm div}(Y)  \frac{\delta \mathfrak{L}}{\delta \sigma} \big) - \big( 
\mathbb{X}^2\ast \frac{\delta \mathfrak{L}}{\delta Y}
+
{\rm div}\mathbb{X}^2 \lrcorner  \frac{\delta \mathfrak{L}}{\delta Y}
\big)\\&\hspace{3cm}-
\sum_{\textbf{k} = 2}^\infty\, \big(
  \mathbb{X}^{\textbf{k+1}}\ast \frac{\delta \mathfrak{L}}{\delta \mathbb{X}^\mathbf{k}}
+ {\rm div} \mathbb{X}^{\textbf{k+1}}  \lrcorner  \frac{\delta \mathfrak{L}}{\delta \mathbb{X}^\mathbf{k}}\big)
\\
\frac{d}{dt}\frac{\delta \mathfrak{L}}{\delta Y}&= - \big( \frac{\delta \mathfrak{L}}{\delta \sigma} d \sigma +  \mathcal{L}_Y (\frac{\delta \mathfrak{L}}{\delta Y}) + 
{\rm div}(Y)\frac{\delta \mathfrak{L}}{\delta Y }\big) \\&\hspace{3cm}-
\sum_{\textbf{k} = 2}^\infty\,\big(  {\rm L}_{\mathbb{X}^{\textbf{k}}} \frac{\delta \mathfrak{L}}{\delta \mathbb{X}^\mathbf{k}} + {\rm div} \mathbb{X}^{\textbf{k}}  \lrcorner  \frac{\delta \mathfrak{L}}{\delta \mathbb{X}^\mathbf{k}} \big),
\\ 
 \frac{d}{dt} \frac{\delta \mathfrak{L}}{\delta \mathbb{X}^\mathbf{2}} &= 
-
\big(\frac{\delta \mathfrak{L}}{\delta Y }\star \sigma\big)  
 -
 \big(\C{L}_Y \frac{\delta \mathfrak{L}}{\delta \mathbb{X}^\mathbf{2}}+ {\rm div}Y  \frac{\delta \mathfrak{L}}{\delta \mathbb{X}^\mathbf{2}} \big)\\&\hspace{3cm}-
\sum_{\textbf{k}=3}^\infty \,\big(
{\rm L}_{\mathbb{X}^{\textbf{k-1}}} \frac{\delta \mathfrak{L}}{\delta \mathbb{X}^\mathbf{k}}+ {\rm div} \mathbb{X}^{\textbf{k-1}}  \lrcorner  \frac{\delta \mathfrak{L}}{\delta \mathbb{X}^\mathbf{k}}\big)
\\
 \frac{d}{dt} \frac{\delta \mathfrak{L}}{\delta \mathbb{X}^\mathbf{m}}&= - \big(
 \C{L}_Y \frac{\delta \mathfrak{L}}{\delta \mathbb{X}^\mathbf{m}} + {\rm div}Y  \frac{\delta \mathfrak{L}}{\delta \mathbb{X}^\mathbf{m}} + \frac{\delta \mathfrak{L}}{\delta \mathbb{X}^\mathbf{m-1}}\star \sigma 
\big) \\&\hspace{3cm}-
\sum_{\textbf{k}=2}^\infty \,\big(
{\rm L}_{\mathbb{X}^{\textbf{k}}} \frac{\delta \mathfrak{L}}{\delta \mathbb{X}^\mathbf{m+k-1}} + {\rm div} \mathbb{X}^{\textbf{k}}  \lrcorner \frac{\delta \mathfrak{L}}{\delta \mathbb{X}^\mathbf{m+k-1}}\big),
\end{split}
\end{equation}
where the last equality is valid when $\mathbf{m}\geq 2$. Recall \eqref{abbri} for coordinate realizations of $\star$ and $\ast$ and \eqref{div-cont} for the contractions involving the divergence operators. 

\subsection{Decomposition of EP Dynamics on the Space Contravariant Tensor Fields}

We are interested in two subdynamics of the Euler-Poincar\'{e} formulation on the space of contravariant tensor fields \eqref{EP-TQ}. One is on the Lie subalgebra $\mathfrak{s}$ in \eqref{s} and the other is on the complementary Lie subalgebra  $\mathfrak{s}$ in \eqref{n}. The former one corresponds to the Euler-Poincar\'{e} formulation of isentropic compressible fluid flow. Let us start to examine this case. We recall the Lie algebra bracket \eqref{pa} on the subalgebra $\mathfrak{s}=\mathcal{F}(\mathcal{Q})\rtimes \mathfrak{X}(\mathcal{Q})$. A direct observation reads that, by choosing $\mathfrak{g}=\mathcal{F}(\mathcal{Q})$ and $\mathfrak{h}=\mathfrak{X}(\mathcal{Q})$, the bracket \eqref{pa} is precisely fitting 
the abstract framework in \eqref{mpla-left}. Here, the left action of $\mathfrak{X}(\mathcal{Q})$ on $\mathcal{F}(\mathcal{Q})$ is the directional derivative. 
So, it determines a left semidirect product algebra where the algebra on $\mathcal{F}(\mathcal{Q})$ is trivial. In this case, the coadjoint action is trivial on $\mathfrak{g}$. Accordingly, we refer to \eqref{mEP-1-left} for the Euler-Poincar\'e equations generated by a Lagrangian function $\mathfrak{L}=\mathfrak{L}(\sigma,Y)$ on $\mathfrak{s}$, as follows
\begin{equation}\label{mEP-1-left-fluid}
 \frac{d}{dt}\frac{\delta \mathfrak{L}}{\delta \sigma}
     = \frac{\delta \mathfrak{L}}{\delta \sigma}\overset{\ast}{\vartriangleleft} Y, \qquad \frac{d}{dt}\frac{\delta \mathfrak{L}}{\delta Y}
     = - \ad_{Y}^{\ast}
\frac{\delta\mathfrak{L}}{\delta Y}
-\mathfrak{b}_{\sigma}^{\ast}\frac
{\delta\mathfrak{L}}{\delta \sigma}.
\end{equation}
To have a more explicit calculation, we need to determine the right hand sides of the equations in \eqref{mEP-1-left-fluid}. A direct computation shows that
\begin{equation}
\begin{split}
\frac{\delta \mathfrak{L}}{\delta \sigma}\overset{\ast}{\vartriangleleft} Y&=-\mathcal{L}_Y \frac{\delta \mathfrak{L}}{\delta \sigma} - \frac{\delta \mathfrak{L}}{\delta \sigma} {\rm div}Y, \\ \ad_{Y}^{\ast}
\frac{\delta\mathfrak{L}}{\delta Y}&= \mathcal{L}_Y \frac{\delta\mathfrak{L}}{\delta Y}+{\rm div}Y\frac{\delta\mathfrak{L}}{\delta Y}
\\
\mathfrak{b}_{\sigma}^{\ast}\frac
{\delta\mathfrak{L}}{\delta \sigma}&=\frac{\delta \mathfrak{L}}{\delta \sigma}d\sigma.
\end{split}
\end{equation}
As a result, the Euler-Poincar\'{e} equations are 
\begin{equation}\label{EP-fl}
\frac{d}{dt}\frac{\delta \mathfrak{L}}{\delta \sigma}=-\mathcal{L}_Y \frac{\delta \mathfrak{L}}{\delta \sigma} - \frac{\delta \mathfrak{L}}{\delta \sigma} {\rm div}Y, \qquad 
\frac{d}{dt}\frac{\delta \mathfrak{L}}{\delta Y}
=-\mathcal{L}_Y \frac{\delta\mathfrak{L}}{\delta Y}-{\rm div}Y\frac{\delta\mathfrak{L}}{\delta Y}-\frac{\delta \mathfrak{L}}{\delta \sigma}d\sigma.
\end{equation}

\smallskip 

Now, we will examine the Euler-Poincar\'{e} 
flow on higher order ($\mathbf{m}\geq 2$) covariant tensor spaces. As mentioned before, the space of symmetric covariant tensors of order  $\mathbf{m}\geq 2$ determines a Lie subalgebra $\mathfrak{n}$ which is depicted in \eqref{n}. This permits us to determine an Euler-Poincar\'{e} flow over $\mathfrak{n}$. Assume a Lagrangian function $\mathfrak{L}=\mathfrak{L}(\sum_{\textbf{k}=2} \mathbb{X}^{\textbf{k}})$ depending on higher order tensors. Then, 
we write 
\begin{equation}\label{EP-hm}
\frac{d}{dt}\frac{\delta \mathfrak{L}}{\delta \mathbb{X}^\mathbf{m}} = - \sum_{\textbf{k} = 2}^\infty\,  \big({\rm L}_{\mathbb{X}^{\textbf{k}}} \frac{\delta \mathfrak{L}}{\delta \mathbb{X}^\mathbf{m+k-1}} + {\rm div} \mathbb{X}^{\textbf{k}}  \lrcorner  \frac{\delta \mathfrak{L}}{\delta \mathbb{X}^\mathbf{m+k-1}} \big),
\end{equation}
where we employed the notations in \eqref{Lie-gen} and \eqref{div-cont}, in a respective order. 

\smallskip 

A direct observation gives that, merely putting together the Euler-Poincar\'{e} equations \eqref{EP-fl} governing the fluid motion  and the Euler-Poincar\'{e} equations \eqref{EP-hm}  governing the covariant tensors of order $\mathbf{m}\geq 2$ will not be equal to the collective motion of the all covariant tensors  fields on $\mathfrak{T}\mathcal{Q}$ exhibited in \eqref{EP-TQ}. The extra terms arising in this coupling are manifestations of the mutual actions of the constitutive subalgebras $\mathfrak{s}$ and $\mathfrak{n}$ on each other. These mutual actions are computed in an abstract framework in \eqref{actions}. Considering the matched pair Euler-Poincar\'{e} equations \eqref{mEP-1}, we first compute the dual actions \cite{EsSu21}
\begin{eqnarray}
(\frac{\delta \mathfrak{L}}{\delta \sigma},\frac{\delta \mathfrak{L}}{\delta Y} )\overset{\ast }{\vartriangleleft}\big(\sum_{\textbf{k}=2} \mathbb{X}^{\textbf{k}}\big)&=&(-\mathbb{X}^{2}\ast \frac{\delta \mathfrak{L}}{\delta Y} - {\rm div}\mathbb{X}^2  \lrcorner   \frac{\delta \mathfrak{L}}{\delta Y},0) \label{dualaction-I}
\\ 
(\sigma,Y)\overset{\ast }{\vartriangleright} \Big( \bigoplus_{\textbf{m}=2}  \frac{\delta \mathfrak{L}}{\delta \mathbb{X}^\mathbf{m}}\Big) &=& \Big ( \C{L}_Y \frac{\delta \mathfrak{L}}{\delta \mathbb{X}^\mathbf{2}}+ {\rm div}Y  \frac{\delta \mathfrak{L}}{\delta \mathbb{X}^\mathbf{2}},   
 \label{dualaction-II}
\\ && \bigoplus_{\textbf{m}=3}^\infty \big( \C{L}_Y \frac{\delta \mathfrak{L}}{\delta \mathbb{X}^\mathbf{m}} + {\rm div}Y  \frac{\delta \mathfrak{L}}{\delta \mathbb{X}^\mathbf{m}} + \frac{\delta \mathfrak{L}}{\delta \mathbb{X}^\mathbf{m-1}}\star \sigma \big )\Big )\nonumber
\end{eqnarray}
where the first one takes values in $\G{s}^*$ and the second takes values in $\G{n}^*$. Further, we compute the cross actions 
\begin{eqnarray}
\G{a}^\ast_{\big(\sum_{\textbf{k}=2} \mathbb{X}^{\textbf{k}}\big)}
 \bigoplus_{\textbf{m}=2}  \Big( \frac{\delta \mathfrak{L}}{\delta \mathbb{X}^\mathbf{m}}\Big)
&=& \Big (-\sum_{\textbf{k}=2}^\infty \big  ( \mathbb{X}^{\textbf{k+1}}\ast \frac{\delta \mathfrak{L}}{\delta \mathbb{X}^\mathbf{k}} + {\rm div}\mathbb{X}^{\textbf{k+1}} \lrcorner \frac{\delta \mathfrak{L}}{\delta \mathbb{X}^\mathbf{k}}\big ), \label{a-ex-dual}  \\ &&
- \sum_{\textbf{k}=2}^\infty \big ( {\rm L}_{\mathbb{X}^{\textbf{k}}} \frac{\delta \mathfrak{L}}{\delta \mathbb{X}^\mathbf{k}} + {\rm div} \mathbb{X}^{\textbf{k}}  \lrcorner  \frac{\delta \mathfrak{L}}{\delta \mathbb{X}^\mathbf{k}}\big )\Big  ),\nonumber 
\\
\G{b}^\ast_{(\s,Y)}(\frac{\delta \mathfrak{L}}{\delta \sigma},\frac{\delta \mathfrak{L}}{\delta Y}) &=& \frac{\delta \mathfrak{L}}{\delta Y} \star \sigma \in \G{T}^*_2\C{Q}.
 \label{b-ex-dual}
\end{eqnarray}
Now, we collect all these terms together and decompose the Euler-Poincar\'{e} equations \eqref{EP-comp-TQ} in the form of matched  Euler-Poincar\'{e} \eqref{mEP-1} as 
\begin{equation}
\begin{split}\label{mEP-1-TQQ}
\frac{d}{dt}(\frac{\delta \mathfrak{L}}{\delta \sigma},\frac{\delta \mathfrak{L}}{\delta Y}) &  =-\ad_{(\s,Y)}^{\ast}%
(\frac{\delta \mathfrak{L}}{\delta \sigma},\frac{\delta \mathfrak{L}}{\delta Y})+ (\frac{\delta \mathfrak{L}}{\delta \sigma},\frac{\delta \mathfrak{L}}{\delta Y})\overset{\ast}{\vartriangleleft}\big(\sum_{\textbf{k}=2} \mathbb{X}^{\textbf{k}}\big) \\&\hspace{2cm} + \G{a}^\ast_{\big(\sum_{\textbf{k}=2} \mathbb{X}^{\textbf{k}}\big)}\bigoplus_{\textbf{m}=2}  \Big( \frac{\delta \mathfrak{L}}{\delta \mathbb{X}^\mathbf{m}}\Big),
\\
\bigoplus_{\textbf{m}=2} \frac{d}{dt} \Big( \frac{\delta \mathfrak{L}}{\delta \mathbb{X}^\mathbf{m}}\Big) &  =-\ad_{\big(\sum_{\textbf{k}=2} \mathbb{X}^{\textbf{k}}\big)}^{\ast}%
 \bigoplus_{\textbf{m}=2}  \Big( \frac{\delta \mathfrak{L}}{\delta \mathbb{X}^\mathbf{m}}\Big)- (\s,Y)\overset{\ast}{\vartriangleright
} \bigoplus_{\textbf{m}=2}  \Big( \frac{\delta \mathfrak{L}}{\delta \mathbb{X}^\mathbf{m}}\Big) \\&\hspace{2cm} - \mathfrak{b}_{(\s,Y)}^{\ast}(\frac{\delta \mathfrak{L}}{\delta \sigma},\frac{\delta \mathfrak{L}}{\delta Y}) . 
\end{split}
\end{equation}
Here, the first terms on the right hand sides are those available in individual EP equations for isentropic fluid flow and the contravariant tensor fields on $\mathbf{m}\geq 2$. The second terms on the right hand sides of both lines are the dual operations in \eqref{dualaction-I} and \eqref{dualaction-II}, respectively. The third terms are cross actions in \eqref{a-ex-dual} and \eqref{b-ex-dual}, respectively.

\subsection{Decomposition of EP Dynamics on Hamiltonian Vector Fields}~

In \eqref{gccl}, we have represented the {\small GCCL} as a map from the space $\mathfrak{T}\mathcal{Q}$ of contravariant tensors to the space $\mathfrak{X}_{\mathrm{ham}}^{\ast }\left( T^{\ast }\mathcal{Q}\right)$ of Hamiltonian vector fields. In Proposition \ref{H_ham-decomp-Pro}, we have used it to arrive at a matched pair decomposition of $\mathfrak{X}_{\mathrm{ham}}^{\ast }\left( T^{\ast }\mathcal{Q}\right)$. This observation and Proposition \ref{phi-prop} motivate us to determine the dual mapping of {\small GCCL} for the matched pair analysis of the Euler-Poincar\'{e} flow on $\mathfrak{X}_{\mathrm{ham}} \left( T^{\ast }\mathcal{Q}\right)$. Let us begin by deriving the explicit realization of the dual mapping of {\small GCCL}  .

Being the dual of a Lie algebra homomorphism, the dual of {\small GCCL} is a momentum and Poisson mapping. Recall the dual space $\mathfrak{X}_{\mathrm{ham}}^{\ast }\left( T^{\ast }\mathcal{Q}\right)$ given in \eqref{X*}.  The dual mapping of {\small GCCL} is computed to be 
\begin{equation} \label{km}
\text{\small GCCL}^*: \mathfrak{X}_{\mathrm{ham}}^{\ast } ( T^{\ast }\mathcal{Q} )\longrightarrow   \mathfrak{T}^{\ast }\mathcal{Q}, \qquad \Pi  \mapsto \bigoplus_{k=0}^{\infty
}\int_{T^{\ast }_p\mathcal{Q}}\left( \theta _{\mathcal{Q}%
}^{k-1}\otimes \vartheta \right) d\textbf{p},
\end{equation}
where $\theta _{\mathcal{Q}}^{k-1}$ is the $\left( k-1\right) $-th
tensor power of the canonical one form $\theta _{\mathcal{Q}}$, and  $\theta^{-1} _{\mathcal{Q}}$ is assumed to be $1$. Here, $
\vartheta $ is a one-form on $T^{\ast }\mathcal{Q}$ given locally by 
\begin{equation}
\vartheta =-\left( k\Pi _{i}+\frac{\partial \Pi ^{j}}{\partial q^{j}}%
p_{i}\right) dq^{i}.
\end{equation}
According to matched pair decomposition of the Lie algebra $\mathfrak{X}_{\mathrm{ham}}^{\ast } ( T^{\ast }\mathcal{Q} )$ as announced in Proposition \eqref{H_ham-decomp-Pro}, we can decompose the dual space  as 
\begin{equation}
\mathfrak{X}_{\mathrm{ham}}^{\ast } ( T^{\ast }\mathcal{Q} )= {\mathfrak{s}}^{c\ast}\oplus {\mathfrak{n}}^{c\ast} \qquad \Pi\mapsto  (\Pi ^{\mathfrak{s\ast}},\Pi^{\mathfrak{n\ast}} ),
\end{equation}
where ${\mathfrak{s}}^{c\ast}$ and ${\mathfrak{n}}^{c\ast} $ are dual spaces of ${\mathfrak{s}}^{c }$ and ${\mathfrak{n}}^{c } $  given in \eqref{subem-s}. The identities in \eqref{dual-inc} read that
\begin{equation}
\text{\small GCCL}^*:{\mathfrak{s}}^{c\ast}\longrightarrow \mathfrak{s}^*,\qquad \Pi ^{\mathfrak{s\ast}} \mapsto (\rho,M)
\end{equation}
where we have 
\begin{equation}
\rho \left( q\right) =-\int_{T_{q}^{\ast }\mathcal{Q}}\frac{\partial \Pi ^{i}%
}{\partial q^{i}}d\textbf{p},\qquad M_{i}(q)=-\int_{T_{q}^{\ast }\mathcal{Q}%
}\left( \Pi _{i}+p_{i}\frac{\partial \Pi ^{j}}{\partial q^{j}}\right) d\textbf{p}
\label{mvtocf-1}
\end{equation}%
where $\rho $ is a real valued function on $\mathcal{Q}$ whereas  $M=M_{i}dq^{i}$ is a differential one-form on $\mathcal{Q}$.  In the density formulation, the first two moments \eqref{mvtocf-1} 
are called plasma-to-fluid map in momentum formulation \cite{MarsdenRatiu-book}. For the complementary dual space 
\begin{equation}
\text{\small GCCL}^*:{\mathfrak{n}}^{c\ast}\longrightarrow \mathfrak{n}^*,\qquad \Pi ^{\mathfrak{n\ast}}\mapsto \mathbf{A}=\bigoplus_{k=2}^{\infty
}\int_{T^{\ast }_p\mathcal{Q}}\left( \theta _{\mathcal{Q}%
}^{k-1}\otimes \vartheta \right) d\textbf{p}.
\end{equation}
 
Now, we rewrite the mathced pair Euler-Poincar\'{e} equation for the present case for a Lagrangian $\mathfrak{L}=\mathfrak{L}(X_h^{\mathfrak{s}},  X_h^{\mathfrak{n}})$ as
\begin{equation}
\begin{split}\label{mEP-1-aa}
\frac{d}{dt}\frac{\delta\mathfrak{L}}{\delta X_h^{\mathfrak{s}}}  &  =
-\ad_{X_h^{\mathfrak{s}}}^{\ast}%
\frac{\delta\mathfrak{L}}{\delta X_h^{\mathfrak{s}}}
+ \frac{\delta\mathfrak{L}}{\delta X_h^{\mathfrak{s}}
}\overset{\ast}{\vartriangleleft} X_h^{\mathfrak{n}} + \mathfrak{a}_{X_h^{\mathfrak{n}}}^{\ast}\frac
{\delta\mathfrak{L}}{\delta X_h^{\mathfrak{n}}},\\
\frac{d}{dt}\frac{\delta\mathfrak{L}}{\delta X_h^{\mathfrak{n}}}  &  =-\ad_{X_h^{\mathfrak{n}}}^{\ast}%
\frac{\delta\mathfrak{L}}{\delta X_h^{\mathfrak{n}}}- X_h^{\mathfrak{s}} \overset{\ast}{\vartriangleright
}\frac{\delta\mathfrak{L}}{\delta X_h^{\mathfrak{n}}} - \mathfrak{b}_{X_h^{\mathfrak{s}}}^{\ast}\frac
{\delta\mathfrak{L}}{\delta X_h^{\mathfrak{s}}}. 
\end{split}
\end{equation}
We wish to determine each of the components in \eqref{mEP-1-aa}. First, we recall the identity \eqref{dual-coad} for the individual coadjoint flows
\begin{equation}
\begin{split}
\text{\small GCCL}^*\circ \ad_{X_h^{\mathfrak{s}}}^{\ast}%
\frac{\delta\mathfrak{L}}{\delta X_h^{\mathfrak{s}}}&=
-\ad^*_{X_{\widehat{(\sigma,Y)}}}\circ \text{\small GCCL}^*\big( \frac{\delta\mathfrak{L}}{\delta X_h^{\mathfrak{s}}}\big),
\\
\text{\small GCCL}^*\circ \ad_{X_h^{\mathfrak{n}}}^{\ast}%
\frac{\delta\mathfrak{L}}{\delta X_h^{\mathfrak{n}}}&=
-\ad^*_{X_{\widehat{\mathbf{X}}}}\circ \text{\small GCCL}^*\big( \frac{\delta\mathfrak{L}}{\delta X_h^{\mathfrak{n}}}\big),
\end{split}
\end{equation}
Further, referring to Lemma \eqref{dual-b-a-}, we have the second terms on the right hand sides of \eqref{mEP-1-aa} as
\begin{equation} \label{dual-triangle-eta-xi-ham}
  \begin{split}
\Big(\text{\small GCCL}^* \big(\frac{\delta\mathfrak{L}}{\delta X_h^{\mathfrak{s}}
}\big)\Big) \overset{\ast }{\vartriangleleft} \mathbf{X} &=\text{\small GCCL}^*\big(
\frac{\delta\mathfrak{L}}{\delta X_h^{\mathfrak{s}}}
 \overset{\ast }{\vartriangleleft} X_h^{\mathfrak{n}}\big), 
\\
(\sigma,Y) \overset{\ast }{\vartriangleright}\text{\small GCCL}^* (\frac{\delta\mathfrak{L}}{\delta X_h^{\mathfrak{n}}}) &=\text{\small GCCL}^*(X_h^{\mathfrak{s}}\overset{\ast }{\vartriangleright}  \frac{\delta\mathfrak{L}}{\delta X_h^{\mathfrak{n}}} \big).
\end{split}
\end{equation}
In addition, we refer to Lemma \eqref{b-a-*} for the third terms on the right hand sides of \eqref{mEP-1-aa} as
  \begin{equation} \label{dual-b-a-ham}
  \begin{split}
\text{\small GCCL}^*\circ \mathfrak{b}^*_{X_h^{\mathfrak{s}}}\frac
{\delta\mathfrak{L}}{\delta X_h^{\mathfrak{s}}}&=-\mathfrak{b}^*_{(\sigma,Y)}\circ \text{\small GCCL}^*\frac
{\delta\mathfrak{L}}{\delta X_h^{\mathfrak{s}}}, 
\\
\text{\small GCCL}^*\circ \mathfrak{a}^*_{X_h^{\mathfrak{n}}} \frac
{\delta\mathfrak{L}}{\delta X_h^{\mathfrak{n}}}&=-\mathfrak{a}^*_{\mathbf{X}}\circ \text{\small GCCL}^*\frac
{\delta\mathfrak{L}}{\delta X_h^{\mathfrak{n}}}.
\end{split}
\end{equation}

As an example to this geometry, and following \cite{holm2009geodesic}, we define a pure quadratic Lagrangian functional on the space $\mathfrak{X}_{\mathrm{ham}}(T^\ast \mathcal{Q})$ of Hamiltonian vector fields as follows
\begin{equation} \label{Lag-geoV}
  \mathfrak{L}(X_h)=\frac{1}{2} \int_{T^\ast\mathcal{Q}} \langle \hat{Q}X_h,X_h\rangle \mu,
\end{equation}
where $\hat{Q}$ is a positive-definite symmetric operator on $\mathfrak{X}_{\mathrm{ham}}(T^\ast \mathcal{Q})$ and $\mu$ is the symplectic volume. It is easy to see that variation of the quadratic Lagrangian $\mathfrak{L}$ with respect to $X_h$ is $\hat{Q}X_h$. We assume that there exist a one-form $\Pi $ so that $\Pi =\hat{Q}X_h$ and the Hamiltonian vector field $X_h$ is written as a sum of two Hamiltonian vector fields $X_h^{\mathfrak{s}}$ and $X_h^{\mathfrak{n}}$ according to the decomposition in (\ref{H_ham-decomp-Pro}). In this decomposition $X_h^{\mathfrak{s}}$ is in the Lie subalgebra $\mathfrak{s}$, whereas $X_h^{\mathfrak{n}}$ is in the Lie subalgebra $\mathfrak{n}$. We further take that the positive-definite symmetric operator $\hat{Q}$ is decomposible as 
\begin{equation}
\hat{Q}=\left(
\begin{array}{cc}
\hat{Q}^{\mathfrak{s}} & 0 
\\ 
0 & \hat{Q}^{\mathfrak{n}}   
\end{array}
\right)
\end{equation}
so that the Lagrangian $\mathfrak{L}$ in \eqref{Lag-geoV} turns out to be a sum
\begin{equation}
\begin{split}
\mathfrak{L}( X_h^{\mathfrak{s}}, X_h^{\mathfrak{n}})&=\mathfrak{L}^{\mathfrak{s}}(X_h^{\mathfrak{s}})+\mathfrak{L}^{\mathfrak{n}}(X_h^{\mathfrak{n}}) \\&=
\frac{1}{2} \int_{T^\ast\mathcal{Q}} \langle \hat{Q}^{\mathfrak{s}} X_h^{\mathfrak{s}},X_h^{\mathfrak{s}} \rangle d\textbf{q}d\textbf{p}
+
\frac{1}{2} \int_{T^\ast\mathcal{Q}} \langle
 \hat{Q}^{\mathfrak{n}} X_h^{\mathfrak{n}},X_h^{\mathfrak{n}}
\rangle d\textbf{q}d\textbf{p},
\end{split}
\end{equation}
see that 
\begin{equation}
\frac{\delta \mathfrak{L}}{\delta X_h}=\left( \frac{\delta \mathfrak{L}^{\mathfrak{s}}}{\delta X_h^{\mathfrak{s}}},
\frac{\delta \mathfrak{L}^{\mathfrak{n}}}{\delta X_h^{\mathfrak{n}}}
 \right)
 =\left(\hat{Q}^{\mathfrak{s}}X_h^{\mathfrak{s}},\hat{Q}^{\mathfrak{n}}X_h^{\mathfrak{n}} \right)=\left(\Pi ^{\mathfrak{s}},\Pi ^{\mathfrak{n}}\right). 
\end{equation}
\section*{Acknowledgments} 

OE is gratefully acknowledge the support of T\"UB\.ITAK (the Scientific and Technological Research Council of Turkey) under the title "Matched pairs of Lagrangian and Hamiltonian Systems"  with the project number 117F426. OE is grateful Prof. Hasan G\"umral and Prof. Serkan S\"utl\"u for many discussions done around the problem addressed in the present work.

\end{document}

%% file: preamble.tex
\usepackage[a4paper,margin=23mm,top=30mm]{geometry}
\usepackage{amsfonts, amsmath, amssymb, amsgen, amsthm, amscd}
\usepackage{newtxtext,newtxmath}
\usepackage[utf8]{inputenc} 

\usepackage{color}

\usepackage[all]{xy}
\usepackage{xcolor}
\usepackage[colorlinks]{hyperref}
\usepackage{mathtools,slashed}
\usepackage{setspace}
\def\s{\sigma}

\def\p{\partial}

\def\ot{\otimes}
\def\odots{\ot\cdots\ot}

\def\ad{\mathop{\rm ad}\nolimits}

\usepackage{enumerate}

\newcommand{\G}[1]{\mathfrak{#1}}
\newcommand{\C}[1]{\mathcal{#1}}
\newcommand{\B}[1]{\mathbb{#1}}

\renewcommand{\geq}{\geqslant}

\numberwithin{equation}{section}

\newtheorem{theorem}{Theorem}[section]
\newtheorem{proposition}[theorem]{Proposition}
\newtheorem{lemma}[theorem]{Lemma}

\theoremstyle{definition}

